\documentclass[10pt,journal,compsoc]{IEEEtran}
\pdfoutput=1
\usepackage{booktabs} 
\usepackage{amsmath,amssymb,amsfonts}
\usepackage{algorithm, algorithmicx}
\usepackage[noend]{algpseudocode}
\usepackage{graphicx}
\usepackage{textcomp}
\usepackage{url}
\usepackage{subfigure}
\usepackage{xcolor}
\usepackage{amsthm}
\usepackage{color}
\usepackage{xcolor}
\usepackage{enumitem}
\usepackage{flushend}
\usepackage{cite}
\usepackage[center]{caption}
\usepackage{soul}

\DeclareMathOperator{\argmax}{argmax}

\newtheorem{definition}{Definition}
\newtheorem{theorem}{Theorem}
\newtheorem{lemma}{Lemma}
\newtheorem{problem}{Problem}

\newcommand{\IMP}{\textsf{IMP}}
\newcommand{\GREEDY}{\textsf{GREEDY}}
\newcommand{\SG}{\textsf{SG}}
\newcommand{\RCELF}{\textsf{RCELF}}
\newcommand{\CELF}{\textsf{CELF}}
\newcommand{\CELFplus}{\textsf{CELF++}}
\newcommand{\PMC}{\textsf{PMC}}
\newcommand{\RIS}{\textsf{RIS}}
\newcommand{\IMM}{\textsf{IMM}}
\newcommand{\TIM}{\textsf{TIM}}
\newcommand{\DSSA}{\textsf{DSSA}}
\newcommand{\IRIE}{\textsf{IRIE}}
\newcommand{\IMRank}{\textsf{IMRank}}
\newcommand{\IC}{\textsf{IC}}
\newcommand{\LT}{\textsf{LT}}
\newcommand{\WC}{\textsf{WC}}
\newcommand{\RR}{\textsf{RR}}
\newcommand{\dblp}{\textsf{DBLP}}
\newcommand{\livejournal}{\textsf{LiveJournal}}
\newcommand{\nethept}{\textsf{NetHEPT}}
\newcommand{\twitter}{\textsf{Twitter}}
\newcommand{\twitterrv}{\textsf{TwitterLarge}}

\newcommand{\G}{$\mathsf{G}$}
\newcommand{\X}{$\mathsf{X}$}
\newcommand{\rG}{$\mathsf{rG}$}

\newcommand{\M}{$\mathsf{M}$}
\newcommand{\Set}{$\mathsf{S}$}
\newcommand{\inu}[1]{$\mathsf{In}(#1)$}
\newcommand{\outu}[1]{$\mathsf{Out}(#1)$}
\newcommand{\I}[1]{$\mathcal{I}($#1$)$}

\newcommand{\stitle}[1]{\vspace*{0.5em}\noindent{\bf #1:\/}}

\newcommand{\trim}{\vspace*{-2mm}}

\usepackage{flushend}

\begin{document}

\title{\RCELF: A Residual-based Approach for Influence Maximization Problem}

\author{Xinxun~Zeng,~~~Shiqi~Zhang,~~~and~~~Bo~Tang
\IEEEcompsocitemizethanks{\IEEEcompsocthanksitem The authors are with the Department
of Computer Science and Engineering, Southern University of Science and Technology, Shenzhen, China. Bo Tang is the corresponding author. \protect\\
Email: \{zengxx@mail,11510580@mail,tangb3@\}sustech.edu.cn}
}

\IEEEtitleabstractindextext{%
\begin{abstract}
Influence Maximization Problem (\IMP{}) is selecting a seed set of nodes in the social network to spread the influence as widely as possible.
It has many applications in multiple domains, e.g., viral marketing is frequently used for new products or activities advertisement.
While it is a classic and well-studied problem in computer science, unfortunately, all those proposed techniques are compromising among time efficiency, memory consumption, and result quality.
In this paper, we conduct comprehensive experimental studies on the state-of-the-art \IMP{} approximate approaches to reveal the underlying trade-off strategies.
Interestingly, we find that even the state-of-the-art approaches are impractical when the propagation probability of the network have been taken into consideration.
With the findings of existing approaches, we propose a novel residual-based approach (i.e., \RCELF{}) for \IMP{},
which i) overcomes the deficiencies of existing approximate approaches, and ii) provides theoretical guaranteed results with high efficiency in both time- and space- perspectives.
We demonstrate the superiority of our proposal by extensive experimental evaluation on real datasets.
\end{abstract}

}

\maketitle

\IEEEdisplaynontitleabstractindextext

\IEEEraisesectionheading{\section{Introduction}} \label{sec:intro}
\IEEEPARstart{S}{ocial} networks (e.g., Facebook, Twitter, Weibo) are becoming an essential media for the public recently.
Viral marketing is widely used in social networks to promote new products or activities.
For example, new products are advertised by some influential users in social networks to other users by ``word-of-mouth" effect.
Therefore, the problem of selecting small but effective influential user set, a.k.a., \emph{Influence Maximization Problem} (\IMP), is the key for successfully viral marketing, and it has been widely studied in literature\cite{kempe2003maximizing,leskovec2007cost,goyal2011celf++,cheng2013staticgreedy,han2018efficient,tang2018online,ohsaka2014fast,borgs2014maximizing,galhotra2016holistic,tang2014influence,tang2015influence,zhou2015upper,nguyen2016stop,huang2017revisiting,kimura2006tractable,chen2009efficient,chen2010scalableA,chen2010scalable,goyal2011simpath,jung2012irie,kim2013scalable,liu2014influence,cohen2014sketch,tang2017influence,nguyen2017importance}.
Mathematically, given a social network \G{}, a positive integer $k$ and a diffusion model \M{},
the \emph{Influence Maximization Problem} (\IMP{}) returns a size-$k$ nodes subset $\mathsf{S}$ (in \G{}) which has the maximum expected influence in \G{}.
The diffusion model \M{} defines the exact ``word-of-mouth" effect manner,
e.g., each influential/activated user can activate its inactive neighbours with a probability in \emph{Independent Cascade (\IC{})} diffusion model.

It is NP-hard to find the optimal size-$k$ set for \IMP{} with \emph{Independent Cascade (\IC{})} and \emph{Linear Threshold (\LT{})} diffusion models~\cite{kempe2003maximizing}.
Due to the hardness of the \IMP{} problem, many approximation approaches ~\cite{leskovec2007cost,goyal2011celf++,cheng2013staticgreedy,ohsaka2014fast,borgs2014maximizing,tang2014influence,tang2015influence,zhou2015upper,nguyen2016stop,huang2017revisiting} and heuristic solutions~\cite{kimura2006tractable,galhotra2016holistic,chen2009efficient,chen2010scalableA,chen2010scalable,goyal2011simpath,jung2012irie,kim2013scalable,liu2014influence,cheng2014imrank,tang2017influence} have been proposed and extensive studied.
However, all existing approaches (cf. Figure~\ref{fig:imp}) are trading-off among time efficiency, memory consumption, and result quality~\cite{arora2017debunking}.
Since the empirical performance of the state-of-the-art approximation approaches (e.g., \IMM{}~\cite{tang2015influence}, \DSSA{}~\cite{nguyen2016stop}) are comparable to, even outperform the state-of-the-art heuristic solutions (e.g., \IRIE~\cite{jung2012irie}, \IMRank~\cite{cheng2014imrank}), we focus on the approximated solutions for \IMP{} in this work.

\begin{figure}
\small
\begin{center}
\includegraphics[width=0.8\columnwidth]{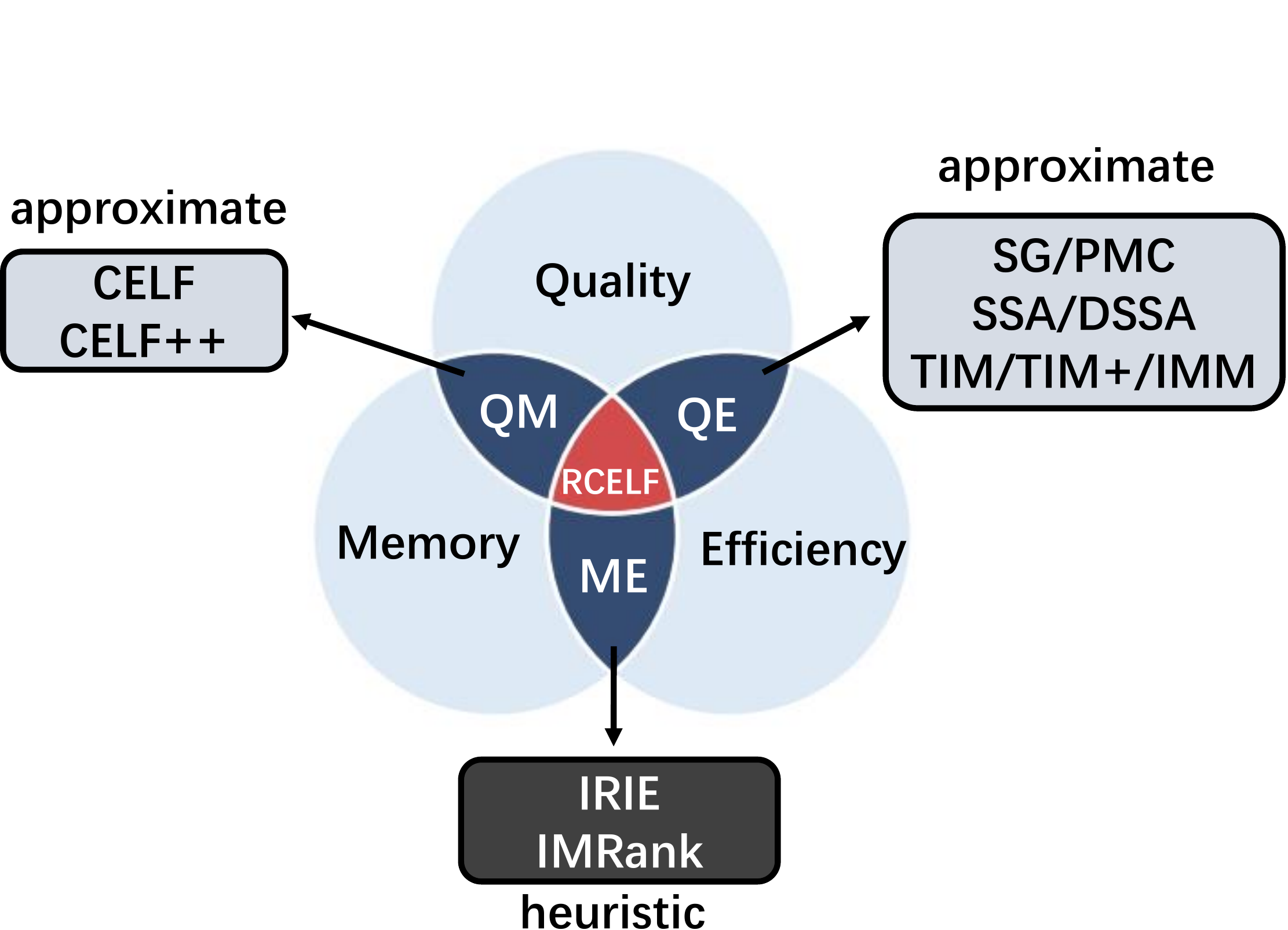}
\caption{\IMP{} solutions (adopted from~\cite{arora2017debunking})} \label{fig:imp}
\end{center}
\vspace{-3mm}
\end{figure}

We classify existing approximation approaches into three categories, \emph{Monte-Carlo Simulation}, \emph{Snapshots} and \emph{Reverse Influence Sampling}, respectively.
We summarize the representative approaches of each category in Figure~\ref{fig:category}.

\stitle{Monte-Carlo Simulation-based approaches}\cite{kempe2003maximizing} is the first to apply Monte-Carlo Simulation techniques to solve \IMP{}.
It achieves $(1-1/e-\epsilon)$-approximation ratio with probability $1-1/n$ if the number of Monte-Carlo simulation times is $\Theta(\epsilon^{-2}k^2n\log(n^2k))$.
\CELF{}~\cite{leskovec2007cost} improves the performance of~\cite{kempe2003maximizing} by exploiting the \emph{submodularity} property of \IMP{} with the same result approximation guarantee.
However, \CELF{} and its variant (i.e., \CELFplus{}~\cite{goyal2011celf++}) are not practical for the large social networks as the cost of Monte-Carlo simulations is rather expensive.
In summary, Monte-Carlo simulation approaches achieve approximation guaranteed results for \IMP{} by incurring extremely expensive time cost.

\stitle{Snapshots-based approaches} They are proposed to improve the time efficiency of the Monte-Carlo Simulation-based approaches.
\SG{}~\cite{cheng2013staticgreedy} is the first to employ snapshots to address \IMP{} problem.
It samples subgraphs \G$_{i}$ (a.k.a., snapshots) of social network \G{} in advance by retaining each edge with a probability of its weight.
The influence of a node is estimated by averaging its influence on all snapshots.
\PMC{}~\cite{ohsaka2014fast} shrinks snapshots in \SG{} into vertex-weighted directed acyclic graphs (DAGs) by using strongly connected components (SCCs) of snapshots as DAGs' nodes.
Through this, it reduces the memory consumption of \SG{}.
Both \SG{} and \PMC{} guarantee $(1-1/e-\epsilon)$-approximation ratio by sampling enough snapshots.
However, the memory overheads of \SG{} and \PMC{} are infeasible when the size of social network \G{} is large.
In conclusion, \emph{snapshots}-based approaches provide approximation guaranteed results for \IMP{} by incurring dramatically high memory consumption (to store snapshots).

\begin{figure}
\small
\begin{center}
    \includegraphics[width=0.8\columnwidth]{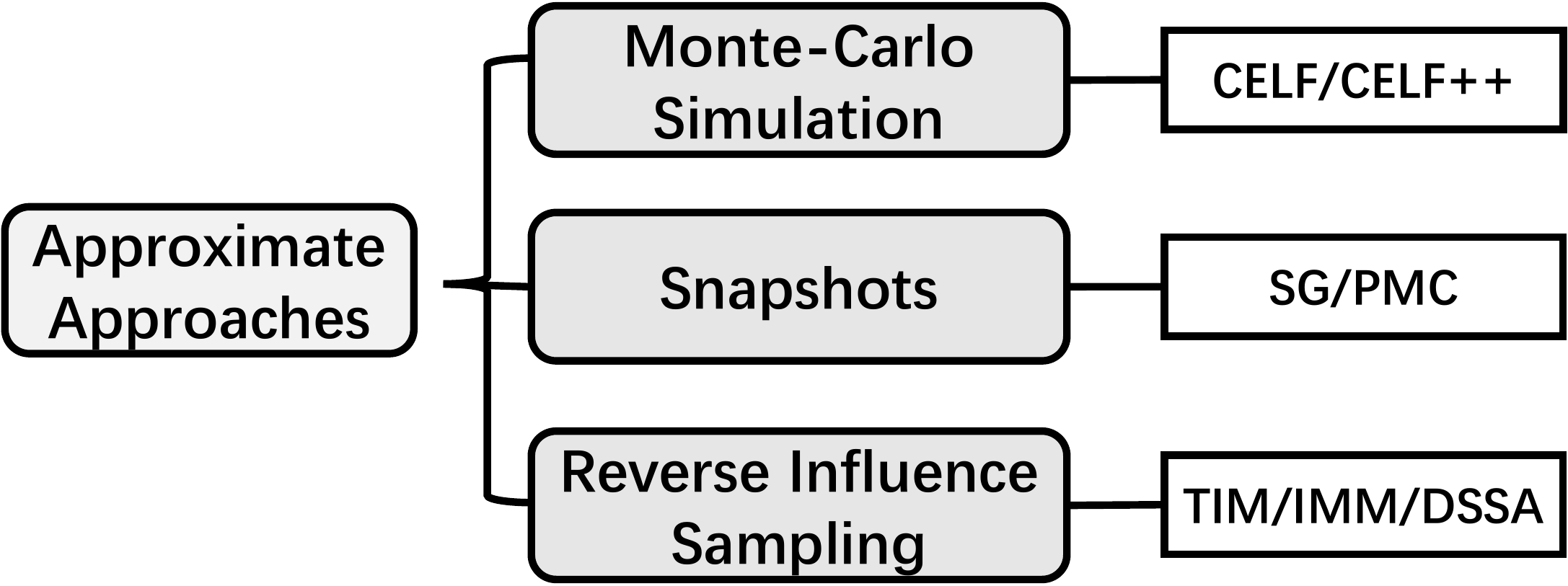}
    \caption{Approximate approach categories for \IMP{}} \label{fig:category}
\end{center}
\end{figure}

\stitle{Reverse Influence Sampling-based approaches}
\cite{borgs2014maximizing} is the first to propose {Reverse Influence Sampling}-based approaches for \IMP{}.
The core idea of reverse influence sampling is: (i) construct reverse reachable sets for the nodes,
(ii) employ a greedy max-coverage algorithm to select the seed node iteratively.
The representative approaches are \TIM{}~\cite{tang2014influence}, \IMM{}~\cite{tang2015influence} and \DSSA{}~\cite{nguyen2016stop}.
All of them guarantee $(1-1/e-\epsilon)$-approximation ratio with a certain number of generated reverse reachable sets.
In general, the memory consumptions of reverse influence sampling-based approaches are much smaller than snapshots-based approaches.
However, it also may be unaffordable as
(i) reverse reachable set size is sensitive to the propagation probability of each edge in the social network,
e.g., \IMM{} incurs 12.3GB memory footprints for 45.6MB \dblp{} dataset when the weighted cascade setting is slightly revised, we will elaborate it in Section~\ref{sec:prob};
(ii) it needs to generate a large number of reverse reachable sets to guarantee the theoretical bound when $\epsilon$ is small~\cite{li2018influence}.
To make matters worse, all reverse reachable sets are stored in the memory for the max-coverage algorithm.

\stitle{Probability assignment in diffusion models}
{In general, diffusion models define how the node can switch its status from \emph{inactive} to \emph{active} on a weighted graph, where the weight of each edge is the influence probability.}
For example, an active node $u$ has single chance to influence its inactive neighbor $v$ with probability $w(u,v)$ in \IC{} model.
In literature, a commonly-used influence probability assignment method in both \IC{} (i.e., \WC{}) and \LT{} models is setting $w(u,v) = 1/|\mathsf{In}(v)|$~\cite{li2018influence}, where $\mathsf{In}(v)$ is the in-degree of $v$.
This assignment method assumes a user probably can be activated if all her incoming neighbors are active in both \WC{} and \LT{} models.
However, it may not be practical in some real-world applications.
For example, many users in Twitter (e.g., the users who use Twitter less than once per day) probably are not be influenced even all the neighbors are active.
Interestingly, there is also another kind of social network where the users can be influenced even only one or few of its neighbors are active,
e.g., users in Pinduoduo\footnote{\url{https://www.pinduoduo.com/}} (an online shopping website in China) can be easily influenced as they can form a shopping team to get a lower price for their purchase.

To overcome the above limitations of common-used probability assignment method, i.e., $1/|$\inu{v}$|$,
We propose a generalized probability assignment method in this work.
Specifically, the probability that node $u$ can activate node $v$ at edge $(u,v)$ is $w(u,v)=\rho/|$\inu{v}$|$, where $\rho$ reflects the activeness of the users.
It is worth to note that the generalized \IC{} and \LT{} models are the conventional \IC{} and \LT{} models when $\rho=1$.

\stitle{Our approach}
In this paper, we propose a novel residual-based algorithm \RCELF{} to overcome the deficiencies of existing approximation approaches.
The core of \RCELF{} is the novel marginal gain computation method based on probability theory.
Specifically, we define the residual capacity of a node as the maximal contribution of that node can make to influence spread value of the seed set.
Initially, the residual capacity of each node is $1$.
During each seed node selection process, the residual capacity of each node will diminish by either being selected as a seed node or being influenced by other selected seed node.
\RCELF{} approach achieves excellent time efficiency as
(i) it enjoys the benefits of the submodularity of the residual-based influence function and cost-effective lazy forward node selection manner,
however, it requires much fewer Monte-Carlo simulations;
(ii) the number of nodes under consideration (i.e., their residual capacities are large than 0) falls quickly during seed set selection process.
and
(iii) two optimizations are devised to speedup \RCELF{}.
Meanwhile, \RCELF{} guarantees $(1-1/e)$-approximation of the result, as elaborated in Section~\ref{sec:ours}.
{From memory consumption perspective,} \RCELF{} only stores the raw social network data. It does not have any extra memory consumption when comparing with snapshot-based approaches and reverse influence sampling-based approaches. Thus, the space complexity of \RCELF{} is optimal, i.e., $O(n+m)$.
{Moreover, \RCELF{} is robust to a generalized probability assignment method, i.e., $\rho/|\mathsf{In}(v)|$, in both \WC{} and \LT{} models. $\rho$ is a tunable parameter and reflects the influence degree of each user in the social network.}
In summary, \RCELF{} achieves excellent time efficiency, low memory consumption and approximation guaranteed result quality for \IMP{} in widely used diffusion models  (i.e., as shown in the center of Figure~\ref{fig:imp}) 


\begin{table}
    \centering
    \small
    \caption{Approximate approaches comparison}\label{tab:approximate}
    \resizebox{0.95\columnwidth}{!}{
        \begin{tabular}{|l|c|c|c|}\hline
       Category & Time  & Memory  & Result  \\
                & Efficiency & Overhead & Quality \\\hline \hline
      {Monte-Carlo Simulation} & Low & Small & High \\ \hline
      {Snapshots} & Median & Large & High \\ \hline
      {Reverse Influence Sampling} & High & Median & High \\ \hline
      {\bf Our Approach (\RCELF{})} & {\bf High} & {\bf Small} & {\bf High} \\ \hline
    \end{tabular}
    }
\end{table}

{
We summarize the comparison among our proposal \RCELF{} and existing approximate approaches for \IMP{} in Table~\ref{tab:approximate}.}
Specifically, the contributions of this paper are summarized as follows.

\begin{itemize}[leftmargin=*]
    \item We conduct comprehensive experiments to reveal the trade-off strategies among the state-of-the-art approximate approaches for \IMP{} (Section~\ref{sec:prob}).
    \item We propose a residual-based algorithm \RCELF{} for \IMP{} to achieve excellent time efficiency, low memory overhead, and approximation guaranteed results concurrently (Section~\ref{sec:ours}).
    \item We evaluate the effectiveness and efficiency of our proposal by extensive experiments on real-world benchmark datasets (Section~\ref{sec:exp}).
\end{itemize}

The remainder of this paper is organized as follows.
Section~\ref{sec:prob} describes the preliminaries and related works of \IMP{} and conducts comprehensive experiments on the state-of-the-art approximate approaches to reveal their underlying issues.
Section~\ref{sec:ours} presents our residual-based approach \RCELF{} for \IMP{}.
Section~\ref{sec:exp} verifies the superiority of our proposal by extensive experiments, followed by the conclusion in Section~\ref{sec:con}.

\section{Influence Maximization Problem}\label{sec:prob}

In this section, we first define the influence maximization problem (\IMP{}) formally.
Then, we conduct extensive preliminary experiments on the representative approximate approaches,
and present the findings of existing approaches.

\subsection{Problem Definition}\label{sec:probdef}

We introduce several fundamental concepts for influence maximization problem (\IMP{}) first.

\begin{definition}(Social Network) \label{def:sn}
A social network is a graph \G$(V,E,W)$, where $V$ $(|V|=n)$ is the set of nodes,
and $E$ is the set of directed edges, $E \subseteq V \times V, (|E|=m)$,
and $W$ is the set of weights of each edge in $E$.
\end{definition}

The weight of edge $(u,v)$ is $w(u,v)$, and $u$ is the incoming neighbor of $v$, vice versa, $v$ is the outgoing neighbor of $u$.
\inu{v} and \outu{v} are the incoming and outgoing neighbor sets of node $v$, respectively.
Given a social network \G{}, the \IMP{} is selecting a small but effective influential user set which could spread the influence in social network \G{} as widely as possible.
We formally define  seed node (i.e., influential user) in Definition~\ref{def:seedset}.

\begin{definition}(Seed Node)\label{def:seedset}
Node $v \in V$ is a seed node if it acts as the source of information diffusion in the social network \G$(V,E,W)$.
The set of seed nodes is called seed set, denoted by \Set.
\end{definition}

Given a social network \G$(V,E,W)$ and seed set \Set{}, the influence of seed set \Set{} in \G{} is the total number of activated nodes with a specified diffusion model \M{}, denoted by \I{\Set}.
\I{\Set} includes both newly activated node during information diffusion process and the initial seed set \Set{}.
The information diffusion process (\M{}) is a stochastic process, the goal of \IMP{} is to maximize the expected influence value, as stated in Problem~\ref{prob:imp}.

\begin{problem}(Influence Maximization Problem, \IMP{})\label{prob:imp}
Given a social network \G$(V,E,W)$, an integer $k$, and diffusion model \M{},
the influence maximization problem \IMP{} is selecting a size-$k$ seed set \Set{} $ \subseteq V$,
such that the expected influence value $\sigma($\Set$)=$ $\mathbb{E}($\I{\Set}$)$ is maximized.
\end{problem}

The information diffusion model \M{} defines the exact information spread manner of seed set \Set{}.
For example, each active user $u$ in step $t$ will active each of its inactive outgoing neighbor $v$ in step $t-1$ with an influence probability $p_{u,v}$
in \emph{Independent Cascade (\IC{})} and \emph{Weighted Cascade (\WC{})} model.
In \emph{Linear Threshold (\LT{})} model,  each edge $(u,v)$ has a weight $p_{u,v}$ and each node $v$ has a threshold $\theta_{v}$.
The node $v$ can be activated if a ``sufficient'' number of its incoming neighbors are active, i.e., $\sum_{v\text{'s active neighbors}~u} p_{u,v} \geq \theta_v$.

\stitle{Influence Probabilities of Edges}
{One of the core components in diffusion model \M{} is determining the influence probabilities/weights of each edge in social network.
The commonly-used influence probability assignment method is weighted cascade (\WC{})~\cite{kempe2003maximizing,cheng2013staticgreedy,tang2014influence,tang2015influence,jung2012irie, chen2010scalableA,arora2017debunking,liu2014influence,galhotra2016holistic,chen2009efficient,li2018influence}.
In particular, all incoming neighbors of $v$ influence $v$ with equal probability $1/|$\inu{v}$|$.}
However, this assignment method ignores the activeness of the users in practical social networks.
e.g., there exists 73\% of Twitter users who use Twitter less than once per day\footnote{\url{http://bit.do/eSyzQ}}.
Such kind of users probably cannot be influenced even all her incoming neighbors are activated.
Interestingly, users can be easily influenced by one or few of her neighbors in other applications.
For example, users in Pinduoduo\footnote{\url{http://bit.do/eSyzg}}, an online shopping website, can invite their friends to form a shopping team to get a lower price for their purchase.
{Thus, an inactive user can be easily influenced by one of her activated friend.}

{In order to overcome the above limitations of common-used probability assignment method, i.e., $1/|$\inu{v}$|$.}
We propose a generalized probability assignment method in this work.
{Specifically, the probability that node $u$ can activate node $v$ at edge $(u,v)$ is $p_{u,v}=\rho/|$\inu{v}$|$, where $\rho$ reflects the activeness of the users.}
The advantage of the generalized probability assignment method is two-fold:
(i) $\rho$ is tunable. It is the conventional setting when $\rho=1$, and it is more general as $\rho$ can be set by users or learnt from training data,
and  (ii) it still enjoys the properties of conventional information diffusion models (e.g., \IC{}, \WC{}, and \LT{}).

\stitle{Properties of \IMP{}}
In order to facilitate the subsequent discussion, we briefly summarize the properties of \IMP{} in this section.

\begin{theorem}[Hardness of \IMP{}]\label{the:hardness}
The problem of influence maximization, as defined in Problem~\ref{prob:imp}, is NP-hard under \IC{} and \LT{} model.
\end{theorem}

In addition to the above hardness of \IMP{}, we present two nice properties of \IMP{}, \emph{monotonicity} and \emph{submodularity} in Theorem~\ref{the:monotone} and~\ref{the:submodular}, respectively.

\begin{theorem}[Monotonicity]\label{the:monotone}
The resulting influence function $\sigma(\cdot)$ is monotone as for any $\mathsf{S}' \subset \mathsf{S}$, we have $\sigma(\mathsf{S}') \leq \sigma(\mathsf{S})$.
\end{theorem}

\begin{theorem}[Submodularity]\label{the:submodular}
For an arbitrary instance of the \IC{} or \LT{} model, the resulting influence function $\sigma(\cdot)$ is submodular.
In other words, for any $\mathsf{S}' \subset \mathsf{S}$ and $v \notin \mathsf{S}$, we have $\sigma(\mathsf{S} \cup v) - \sigma(\mathsf{S}) \leq \sigma(\mathsf{S}' \cup v) - \sigma(\mathsf{S}')$.
\end{theorem}

The marginal gain of node $v$ w.r.t. seed set $\mathsf{S}$ is $mg(v | \mathsf{S}) = \sigma(\mathsf{S} \cup v) - \sigma(\mathsf{S})$.
We omit the proofs of Theorem~\ref{the:hardness},~\ref{the:monotone} and~\ref{the:submodular}, and refer interested reader to~\cite{kempe2003maximizing}.

\begin{figure*}
  \small
  \centering
  \begin{tabular}{cc}
    \includegraphics[width=0.280\columnwidth]{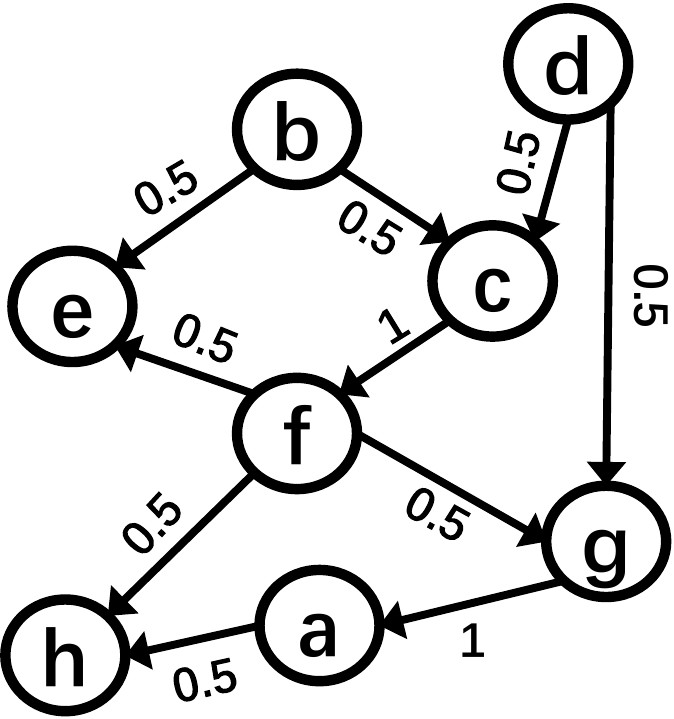}
    &
    \includegraphics[width=1.55\columnwidth]{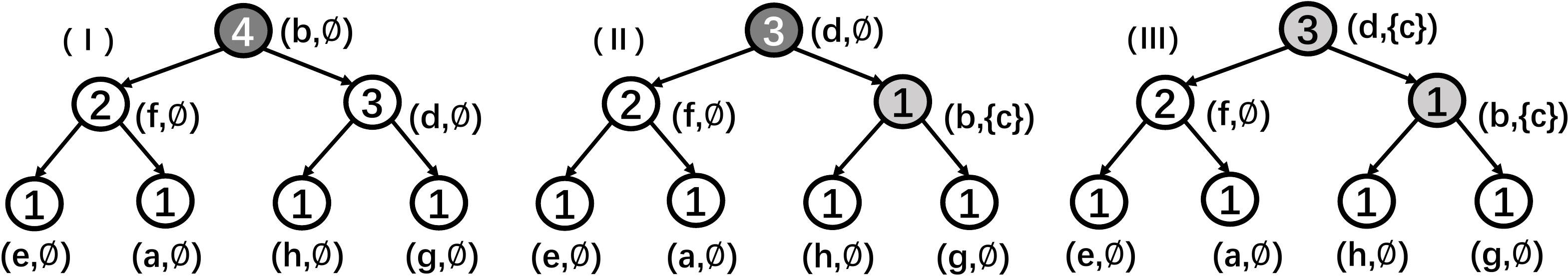} \\
    (a) Social network \G{} & (b) \CELF{} idea illustration
  \end{tabular}
  \caption{Monte-Carlo simulation illustration example}\label{fig:mc}
\end{figure*}

\subsection{Approximate Approaches for \IMP{}}\label{sec:eximp}
Due to the hardness to find the optimal solution for \IMP{} (cf. Theorem~\ref{the:hardness}), a plethora of techniques~\cite{leskovec2007cost,goyal2011celf++,cheng2013staticgreedy,ohsaka2014fast,borgs2014maximizing,tang2014influence,tang2015influence,zhou2015upper,nguyen2016stop,huang2017revisiting} have been proposed to \IMP{} with theoretical approximate bound.
In this section, we briefly introduce the key ideas of each category of approximate approaches.

For each approach, we conduct extensive preliminary experiments and present the experimental findings to reveal the underlying issues.

\stitle{Monte-Carlo Simulation-based \GREEDY{} and \CELF{}} \GREEDY{}~\cite{kempe2003maximizing} is the first approach which employs Monte-Carlo simulation method to address \IMP{}.
The sketch of \GREEDY{} is shown in Algorithm~\ref{alg:greedy}.
\GREEDY{} selects the node which has the largest marginal gain by Monte-Carlo simulation (Line~\ref{alg:every}) during each node selection iteration.
In order to reduce the pain from unguaranteed submodularity~\cite{cheng2013staticgreedy} during Monte-Carlo simulations,
\GREEDY{} runs Monte-Carlo simulation $r$ times for each node $v$, typically, $r$ is 10,000 or 20,000~\cite{kempe2003maximizing,leskovec2007cost,cheng2013staticgreedy}.
Thus, the time cost of \GREEDY{}  is extremely expensive. 

\begin{algorithm}
    \caption{Greedy($\mathsf{G}(V,E,W)$, $k$, \M{})} \label{alg:greedy}
    \begin{algorithmic}[1]
    \State Initialize seed set $\mathsf{S} \leftarrow \emptyset$
    \State $i \leftarrow 1$
    \For{ $i$ from 1 to $k$} \label{alg:iteration}
        \State $\mathsf{S} \leftarrow \mathsf{S} \cup \argmax_{\forall v \in V}\{\sigma(\mathsf{S} \cup v) - \sigma(\mathsf{S})\}$ under \M{} \label{alg:every}
    \EndFor
    \State Return $\mathsf{S}$
    \end{algorithmic}
\end{algorithm}

\CELF{}~\cite{leskovec2007cost} is devised to improve the time efficiency of \GREEDY{}.
It exploits the submodularity of \IMP{} (cf. Theorem~\ref{the:submodular}) to reduce a lot of unnecessary marginal gain computations.
Consider social network \G{} in Figure~\ref{fig:mc}(a), \CELF{} finds the node with the largest marginal gain at the beginning (i.e., $\mathsf{S}=\emptyset$).
\CELF{} maintains a max-heap for the marginal gains of each node w.r.t. seed set $\mathsf{S}$.
Seed set $\mathsf{S}$ is $\{c\}$ after the first seed selection iteration as $mg(c | \emptyset)$ is the largest.
\CELF{} maintains the rest max-heap by removing node $c$, as illustrated in Figure~\ref{fig:mc}(b-I).
At the second iteration, \CELF{} gets the root of the max-heap $mg(b | \emptyset)=4$.
 \CELF{} computes node $b$'s marginal gain with seed set $\mathsf{S} = \{ c \}$, i.e., $mg(b | \{c\})=1$, and updates the max-heap accordingly (cf. Figure~\ref{fig:mc}(b-II)).
The max-heap root turns to $mg(d | \emptyset )= 3$, \CELF{} then updates it to $mg(d | \{ c \} )= 3$.
For any descendant of the root in max-heap, their marginal gains must be smaller than $3$ due to the submodularity of \IMP{}.
Thus, $d$ has the largest marginal gain with seed set $\{c\}$, and it is selected at the second seed selection iteration, i.e., $\mathsf{S}=\{c,d\}$.
In summary, \CELF{} works in a lazy manner.
It only computes the marginal gain of node $v$ with the latest seed set $\mathsf{S}$ when it is necessary,
e.g., \CELF{} only computes the marginal gains of $b$ and $d$ at the second iteration in the above example,
{it reduces lots of unnecessary marginal gain computations.}
Hence, \CELF{} is faster than \GREEDY{} many orders of magnitude (i.e., 700 times~\cite{leskovec2007cost})
However, \CELF{} is still not feasible to large networks, e.g., a social network with 1 million nodes.

\stitle{Experimental evaluations and findings}
{ We test \CELF{} on \textsf{NetHEPT} with 15K nodes and 62K edges, and it does not return the size-50 seed set within 30 hours.}
Monte-Carlo simulation based approaches (i.e., \GREEDY{} and \CELF{}) incur expensive computation time and low memory consumption,
and provide theoretical guaranteed approximate solutions for \IMP{}.

\begin{figure}
  \centering
  \begin{tabular}{cc}
    \includegraphics[width=0.30\columnwidth]{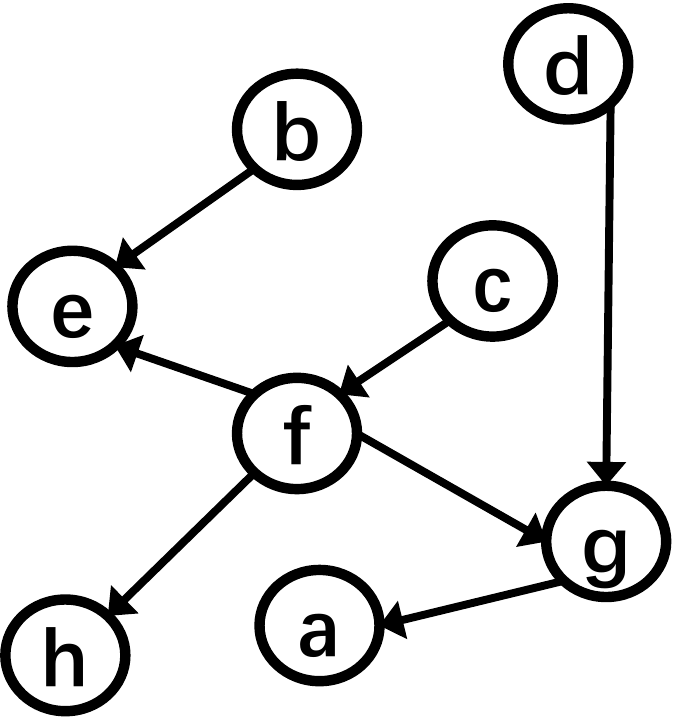}
    ~~~
    &
    ~~~
    \includegraphics[width=0.30\columnwidth]{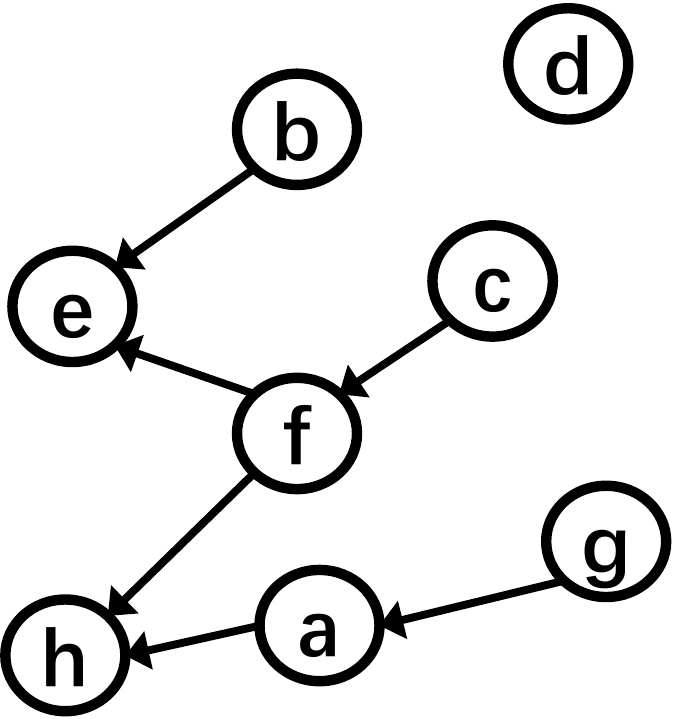} \\
    (a) Snapshot $\mathsf{G'_1}$ & (b) Snapshot $\mathsf{G'_2}$
  \end{tabular}
   
  \caption{Snapshots generation by coin-flip technique}\label{fig:snapshot}
   
\end{figure}

\stitle{Snapshots-based \SG{} and \PMC{}}
Snapshots-based approaches (\SG{}~\cite{cheng2013staticgreedy} and \PMC~{}\cite{ohsaka2014fast}) are proposed to improve the time efficiency of Monte-Carlo simulation based approaches.
\SG{} is the first approach that applies snapshots idea to address \IMP{}.
Instead of running lots of Monte-Carlo simulations in~\GREEDY{} and \CELF{}, \SG{} samples $r$ snapshots of input social network \G{} by coin flip technique.
\SG{} flips all coins with bias $p_{u,v}$ to produce several snapshots in advance,
e.g., Figure~\ref{fig:snapshot}(a) and (b) are the snapshots of the original social graph in Figure~\ref{fig:mc}(a).
\SG{} selects seed nodes iteratively by (1) computing the marginal influence of each node $v$ by averaging the total reachable nodes of $v$ in all snapshots,
(2) selecting the node with the largest average marginal influence, and (3) remove the node and all its reachable nodes in all snapshots.
For example, In Figure~\ref{fig:snapshot}, node $c$ has largest average reachable nodes (i.e., 5) as its reachable nodes in snapshots $\mathsf{G}'_1$  and $\mathsf{G}'_2$ are $\{c,f,e,g,a,h\}$ and $\{c,f,e,h\}$, respectively.
Then \SG{} selects node $c$ and removes its reachable nodes in Figure~\ref{fig:snapshot}(a) and (b) before the second iteration.

\PMC{} improves \SG{} by reducing the memory consumption overhead of the snapshots in \SG{}.
Particularly, \PMC{} generates the Directed Acyclic Graph (DAG) of each snapshot by identifying the strongly connected components (SCC) in it.
However, the space consumption improvement extent of \PMC{} depends on the connectivity of original graph and its snapshots.
\SG{} and \PMC{} guarantee $(1-1/e-\epsilon)$ approximation ratio as the proofs in~\cite{li2018influence} and~\cite{ohsaka2014fast}, respectively.

\begin{figure} 
\small
\begin{center}
		\includegraphics[width=0.6\columnwidth]{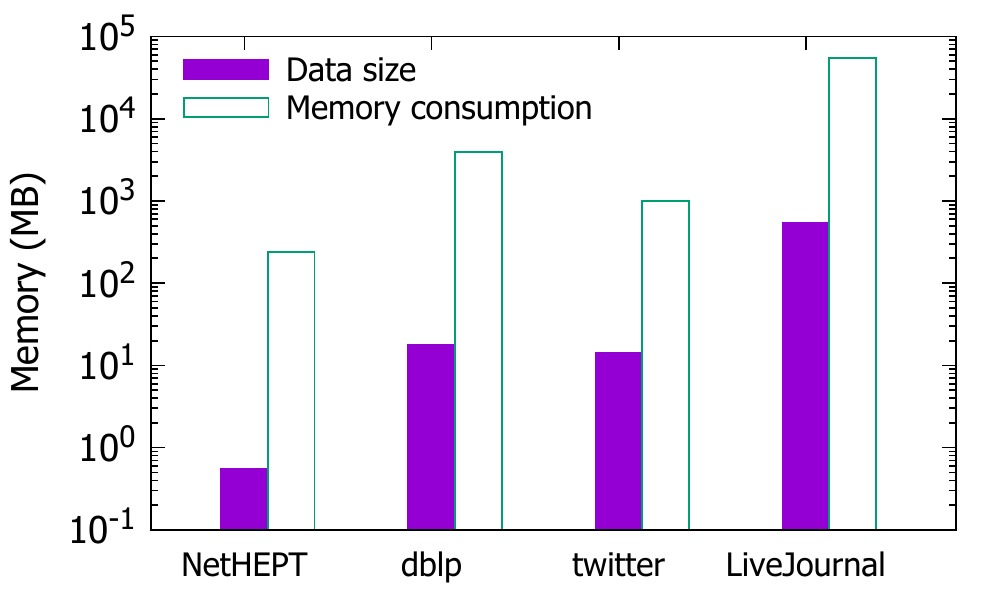}
        
\caption{\PMC{} memory consumption vs. raw data size} \label{fig:pmc_mem}
\end{center}
 
\end{figure}

\stitle{Experimental evaluations and findings}
We run \PMC{} by setting $r=200$ ~\cite{cheng2013staticgreedy} on four benchmark datasets (cf. Section~\ref{sec:setting}).
Figure~\ref{fig:pmc_mem} shows the memory consumption of \PMC{} and the raw data size.
The memory consumption of \PMC{} is 55X to 214X of raw data size.
For example, the size of \textsf{LiveJournal} is 0.5G, its \PMC{} memory consumption is 42.9G.
It is unaffordable for large even median social networks in commodity PCs with 16G or 32G memory.
Snapshots-based approaches (i.e., \SG{} and \PMC{}) achieve good time efficiency by incurring huge memory consumption,
and provide approximate ratio guaranteed solutions for \IMP{}.

\stitle{Reverse Influence Sampling-based \IMM{} and \DSSA{}}
Borgs et al.~\cite{borgs2014maximizing} is the first to propose reverse influence sampling (\RIS{}) method for \IMP{} under the \IC{} and \WC{} model.
The core concept in \RIS{} is reverse reachable set (\RR{} set).
Formally, the \RR{} set of node $v$ is the set of nodes in \G{} that can reach $v$, i.e., $\forall u \in \mathsf{RR}(v)$, there is a path from $u$ to $v$ in \G{}.
\RIS{} method includes two phases: (i) \RR{} sets generation phase, and (ii) node selection phase.
Consider that we run \RIS{} on the social network \G{} in Figure~\ref{fig:mc}(a).
For the \RR{} set generation phase, we first transpose \G{} in Figure~\ref{fig:mc}(a) to $\mathsf{G}^T$ in Figure~\ref{fig:rrset}(a).
\RIS{} randomly picks a node in $\mathsf{G}^T$ and run Monte-Carlo simulation from it to generate its reachable set, i.e., $\mathsf{RR}(e)=\{e,f,c\}$.
\RIS{} repeats the above procedure several times to generate the \RR{} sets, as shown in Figure~\ref{fig:rrset}(b).
For node selection phase, \RIS{} solves the max-coverage problem~\cite{khuller1999budgeted} to select $k$ nodes to cover the maximum number of generated \RR{} sets in above phase.
For example, node $c$ is selected as it covers the maximum number of \RR{} sets (i.e., 3) in Figure~\ref{fig:rrset}(b).
The reverse reachable set which covers node $c$ is marked to be ignored in subsequent seed selections.
Theoretically, \RIS{} returns $(1-1/e-\epsilon)$-approximation result~\cite{borgs2014maximizing} with at least a constant probability if the total examined number of nodes and edges reaches a pre-defined threshold $\tau$, which reflects the number of generated \RR{} sets indirectly.

\begin{figure}
  \centering
  \begin{tabular}{cc}
    \includegraphics[width=0.25\columnwidth]{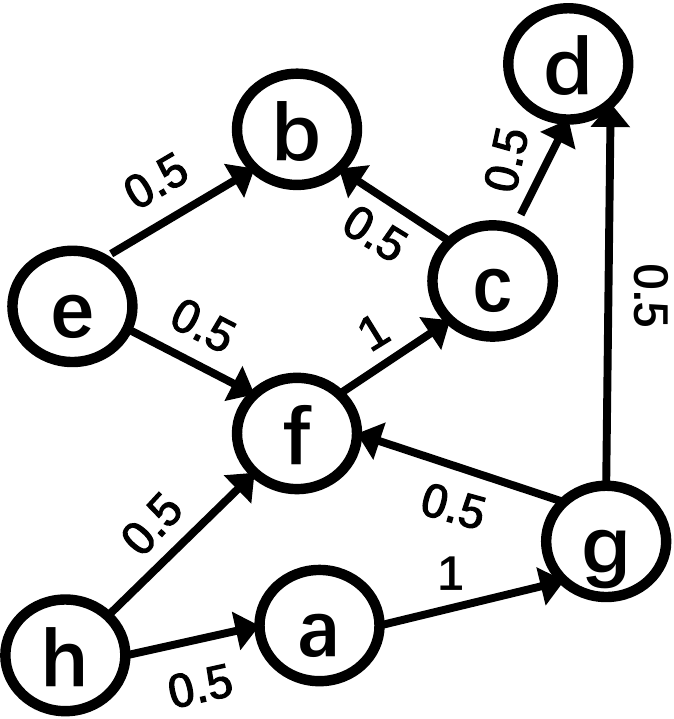}
    &
    \includegraphics[width=0.30\columnwidth]{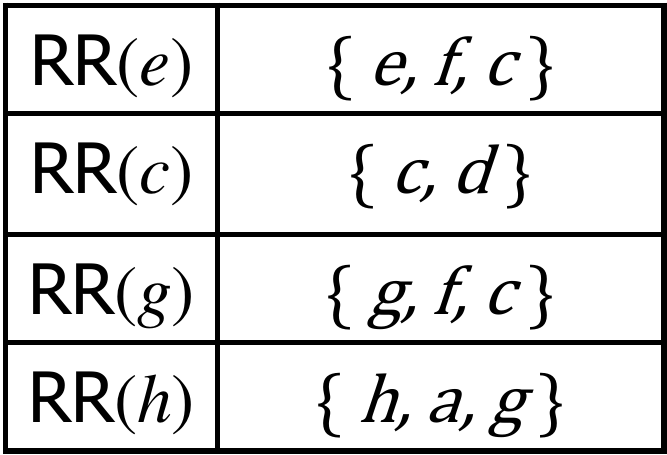} \\
    (a) Transpose graph (\G$^T$) & (b) Reverse reachable (\RR{}) sets
  \end{tabular}
  
  \caption{\RIS{} method illustration}\label{fig:rrset}
   
\end{figure}

Since there is a large hidden constant factor in the asymptotic time complexity of \RIS{}, which bound the practical efficiency of \RIS{}.
In order to address that, \cite{tang2014influence} proposed \emph{Two Phase Influence Maximization} (a.k.a., \TIM{}),
which returns $(1-1/e-\epsilon)$-approximation solution with at least $(1-n^l)$ probability, and it runs in $O((k+l)(m+n) \log n/\epsilon^2)$ times.
\TIM{} samples pre-decided $\theta$ \RR{} sets, instead of using threshold $\tau$ on computation cost to indirectly control the number in \RIS{}.
Later, \IMM{} exploits a classical statistical tool (martingales~\cite{williams1991probability}) to improve the parameters estimation phase in \TIM{}.
Since the number of generated \RR{} sets can be arbitrarily larger than theoretical thresholds $\theta$ in \TIM{} and \IMM{},
Nguyen et al.~\cite{nguyen2016stop} (i) unify the necessary sampled \RR{} sets size in~\cite{borgs2014maximizing,tang2014influence,tang2015influence} to guarantee $(1-1/e-\epsilon)$-approximation ratio,
and (ii) propose \DSSA{} to achieve the minimum number of \RR{} set samples.
Technically, \IMM{} and \DSSA{} adopt a bootstrap strategy to probe the sampling size of \RR{} sets.
The procedure is: (1) initialize $r$ \RR{} sets based on a given formula;
(2) select a size-$k$ seed set \Set{} by max-coverage algorithm;
(3) evaluate the coverage ratio of \Set{}.
If the coverage ratio is under the stopping condition, increase \RR{} sets size and repeat (2) and (3).
Otherwise, terminate and return \Set{}.
The time efficiency and memory consumptions of \IMM{} and \DSSA{} heavily depend on the number of generated \RR{} sets.
Theoretically, the number of generated \RR{} sets is decided by two parameters:
(i) $\epsilon$, a large number of \RR{} sets will be generated to guarantee the theoretical bound when $\epsilon$ is small~\cite{li2018influence};
and (ii) $\rho$ (i.e., the influence probability of each edge $\rho/|$\inu{v}$|$).
Specifically, \IMM{} and \DSSA{} perform pretty good in conventional \WC{} model (i.e., $\rho=1$) as the expected number of node $v$'s direct reverse reachable neighbors is 1 as the influence probability from node $u$ to $v$ is $1/|$\inu{v}$|$.
However, their running times increase dramatically when $\rho$ scale up to 1.5~\cite{tang2017influence}.

\begin{figure}
  \centering
  \begin{tabular}{cc}
    \includegraphics[width=0.45\columnwidth]{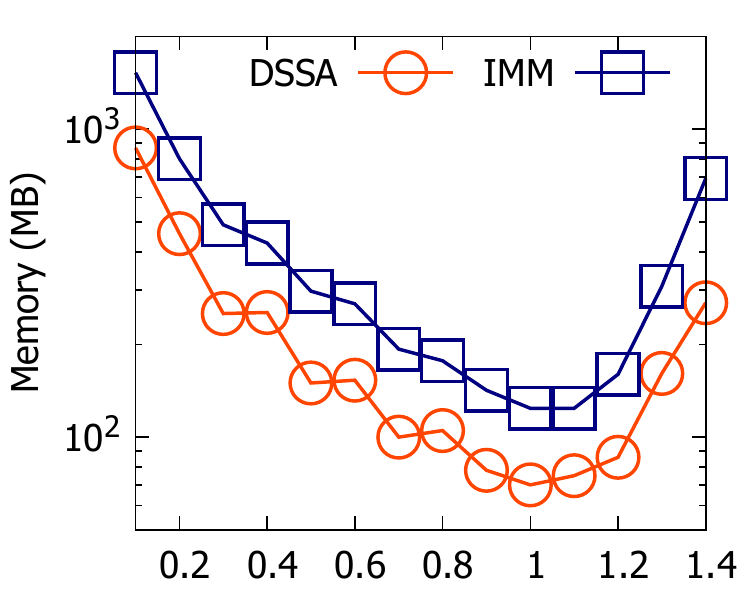}
    &
    \includegraphics[width=0.45\columnwidth]{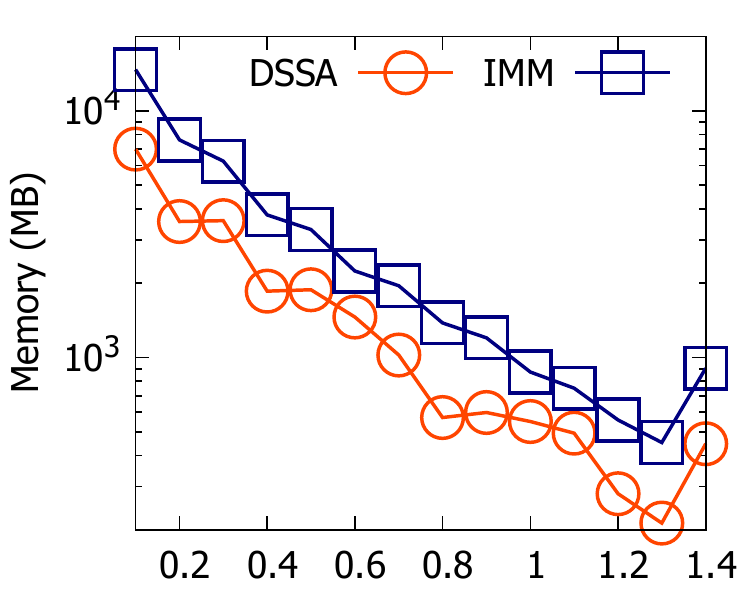}
    \\
    (a) In \textit{Twitter}& (b) In \textit{dblp}
    \\
    \includegraphics[width=0.45\columnwidth]{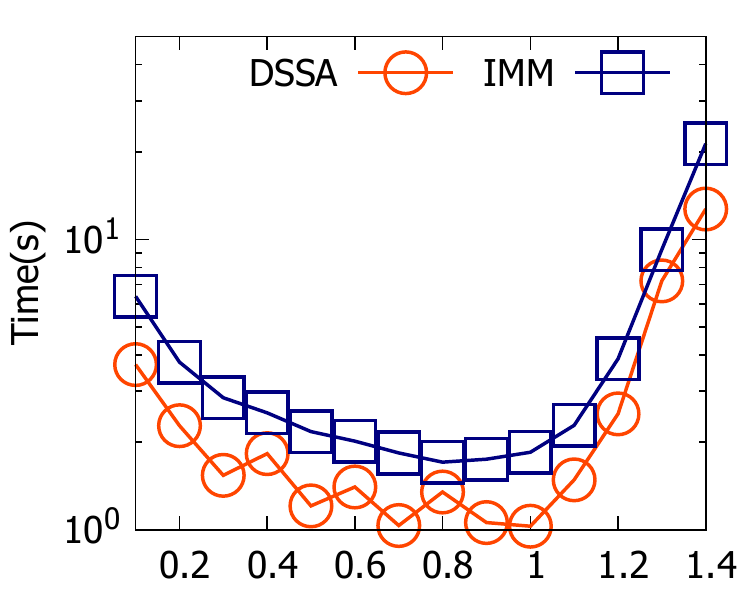}
    &
    \includegraphics[width=0.45\columnwidth]{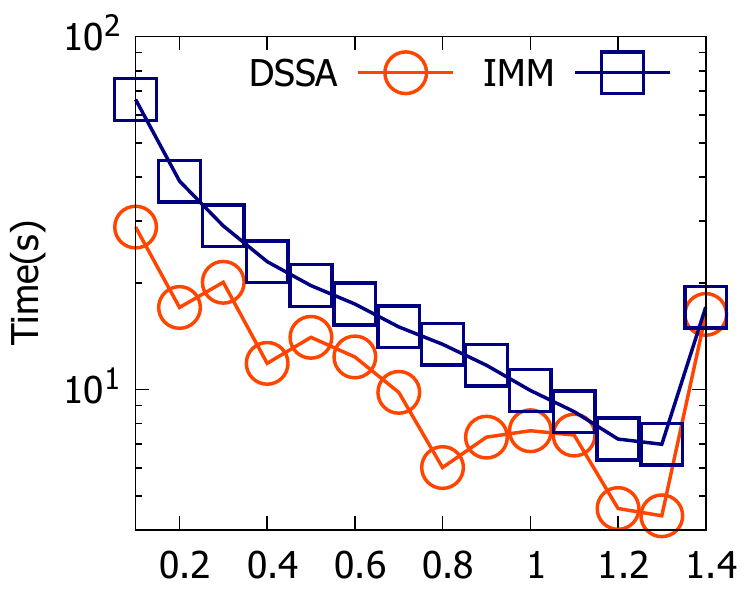}
    \\
    (c) In \textit{Twitter}& (d) In \textit{dblp}
  \end{tabular}
  \caption{\IMM{} and \DSSA{} evaluation by varying $\rho$}\label{fig:ris_mem}
   
\end{figure}

\stitle{Experimental evaluations and findings}
We evaluate the performance of \IMM{} and \DSSA{} methods in \WC{} model by varying the influence probability $\rho/|$\inu{v}$|$ on \twitter{} and \dblp{}.
The memory consumption of both \IMM{} and \DSSA{} are unaffordable when $\rho$ is scaling up (or down) as illustrated in Figure~\ref{fig:ris_mem}(a) and (b).
For example, the memory consumption of \IMM{} with $\rho=0.1$ is almost 12.3X and 16.8X over the cost of \IMM{} with $\rho=1.0$ in \textit{Twitter} and \textit{dblp}, respectively.
The time costs of \IMM{} and \DSSA{} by varying $\rho$ are shown in Figure~\ref{fig:ris_mem}(c) and (d).
Obviously, both approaches are degenerating seriously when scaling up or down the influence probability in each edge $(u,v)$, i.e., $w(u,v) = \rho/|$\inu{v}$|$.
Both \IMM{} and \DSSA{} perform pretty good in terms of time efficiency and memory consumption in common-used probability assignment method (i.e., $1/|$\inu{v}$|$ ) in diffusion models.
The reason is reverse influence sampling (\RIS{}) technique exploits the expected number of node $v$'s direct reverse reachable neighbors is 1 in conventional diffusion models implicitly.
The running time and memory consumption of \IMM{} and \DSSA{} are sensitive to the propagation probabilities (i.e., $\rho/|$\inu{v}$|$).
Both \IMM{} and \DSSA{} are impractical when $\rho$ scales up (or down), as the results shown in Figure~\ref{fig:ris_mem}.
Specifically, when $\rho < 1$, \IMM{} and \DSSA{} require more bootstrap iterations, however, each iteration generates double \RR{} sets.
For $\rho > 1$, the number of node in each generated \RR{} set by \IMM{} and \DSSA{} will increase dramatically as the strongly connected properties of the social network.

\subsection{Other Related Works}\label{sec:rel}
Influence maximization problem \IMP{} is first solved in algorithmic perspective by probability\cite{domingos2001mining}.
Beyond above discussed approximate approaches, there are many heuristic-based approaches~\cite{kimura2006tractable,galhotra2016holistic,chen2009efficient,chen2010scalableA,chen2010scalable,goyal2011simpath,jung2012irie,kim2013scalable,liu2014influence, tang2017influence}.
We omit the details here and refer the interested readers to a recent survey~\cite{li2018influence}.
Very recently, several works~\cite{han2018efficient,tang2018online} are proposed for \IMP{} variants (e.g., online and adaptive \IMP{}),
we skip the discussion as the scope of this work is conventional \IMP{}. 

\section{Residual-based Approach}\label{sec:ours}
Existing approximate approaches for \IMP{} are compromising either time efficiency or memory overhead for result quality.
In this section, we propose a novel residual-based approach (i.e., \RCELF{}) for \IMP{} to overcome this dilemma.
We present the fundamental concepts of \RCELF{} approach in Section~\ref{sec:residual}.
In Section~\ref{sec:rcelfopt}, we describe the backbone of \RCELF{} and devise two performance optimization techniques for it.
We conduct correctness, complexity and approximate analysis of \RCELF{} in Section~\ref{sec:ana}.

\subsection{\RCELF{} Approach}\label{sec:residual}
Generally, each node $v \in V$ in social network contributes to the influence spread value of seed set \Set{},
i.e., $\sigma($\Set{}$)$, by either being selected as a seed node or being influenced by other seed nodes.
In this work, we propose a novel concept, \emph{node residual capacity}, to capture the contribution of each node to the influence spread value.
\RCELF{} selects the node $v \in V - \mathsf{S}$ with the largest marginal gain (based on node residual capacity) as a seed node at each iteration..
The residual capacity of each node diminishes during the seed node selection process.
In order to capture the contribution of each node, we define residual-based social network formally in Definition~\ref{def:rgraph}.

\begin{definition}[Residual based Social Network]~\label{def:rgraph}
Residual based social network \rG$(V,E,W,C)$ is a social graph \G$(V,E,$ $W)$ (cf. Definition ~\ref{def:sn}) with residual capacity set $C$.
Initially, the residual capacity of each node $v \in V$ in $C$ is $\mathsf{RC}(v)=1$.
\end{definition}

Given a diffusion model, the core subroutine of \IMP{} is marginal gain computation.
Given seed set $\mathsf{S}$, the marginal gain of node $u$ is computed by $mg(u|\mathsf{S}) = \delta(\mathsf{S} \cup \{u\}) - \delta(\mathsf{S})$ in literature.
In this work, \RCELF{} computes the marginal gain $mg(u|\mathsf{S})$ by exploiting the residual capacity of every nodes during each node selection iteration,
i.e., the contribution of each node $v \in V - \mathsf{S}$ to $mg(u| \mathsf{S})$.

Given residual based social graph \rG$(V,E,W,C)$ and seed set $\mathsf{S}$.
The marginal gain $mg(u | \mathsf{S})$ is contributed by two parts:
(i) node $u$, and (ii) the nodes which can be influenced by node $u$.
Intuitively, the contribution of node $u$ to $mg(u | \mathsf{S})$ is its residual capacity $\mathsf{RC}(u)$.
The contribution of other nodes ($i.e., \forall v \in V - (\mathsf{S} \cup u)$) to $mg(u | \mathsf{S})$ is a bit more intricate.
In subsequent sections, we present the marginal gain computations of \RCELF{} with \LT{} and \IC{} model, respectively.

\begin{figure}
  \centering
  \begin{tabular}{cc}
    \includegraphics[width=0.35\columnwidth]{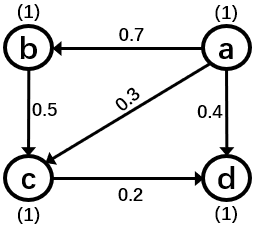}
    ~~~~
    &
    ~~~~
    \includegraphics[width=0.35\columnwidth]{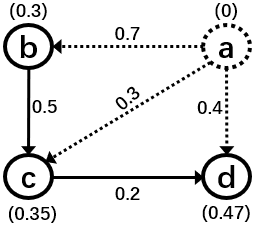}\\
    (a) $\mathsf{S}=\emptyset$ & (b) $\mathsf{S}=\{a\}$
  \end{tabular}
  \caption{\RCELF{} illustration}\label{fig:ltrsd}
\end{figure}

\subsubsection{\RCELF{} Marginal Gain Computation in \LT{} Model}  
In \LT{} model, each node $v$ uniformly chooses a threshold $\theta_v$ from the range $[0,1]$.
It can be activated if $\sum_{v\text{'s activated neighbor} u} w(u,v) \geq \theta_{v}$. 
Take Figure~\ref{fig:ltrsd}(a) as an example, 
the probability of node $b$ can be activated by node $a$ is $\mathsf{Pr}[w(a,b) \geq \theta_{b}] = \mathsf{Pr}[0.7 \geq \theta_{b}] = 0.7$, as $\theta_b$ is uniformly chooses from $[0,1]$.
In $\LT{}$ model, $\forall u \in V, \theta_u$ is independent and identically distributed.
Hence, given a path $\mathsf{P} = \{v_1, v_2,$ $\cdots, v_m \}$, 
the probability of node $v_m$ can be activated by node $v_1$ is  $\mathsf{Pr}_{v_1}(\mathsf{P})= \prod_{i=1}^{m-1}w(v_i, v_{i+1})$.
Definition~\ref{def:ltpathprob} defines the active / influence probability in \LT{} model formally.

\begin{definition}[Influence Probability in \LT{}]\label{def:ltpathprob}
Given residual network \rG$(V,E,W,C)$, the probability of $v_1$ activates $v_m$ through path $\mathsf{P} = \{v_1, v_2, \cdots, v_m \}$ is $\mathsf{Pr}_{v_1}(\mathsf{P})= \prod_{i=1}^{m-1}w(v_i, v_{i+1})$.
 Set $\mathbb{P}$ includes all paths from node $v_1$ to $v_m$ in \rG, 
the overall probability that node $v_{1}$ influences node $v_{m}$ is $\mathsf{Pr}(v_1,v_m)= \sum_{\mathsf{P} \in \mathbb{P}}\mathsf{Pr}_{v_1}(\mathsf{P})$.
\end{definition}

\stitle{Example} In Figure~\ref{fig:ltrsd}(a), node $a$ can reach node $d$ in three paths, i.e.,
$\mathsf{P}_{1}=\{a, d\}$, $\mathsf{P}_{2}=\{a, c, d\}$ and $\mathsf{P}_{3}=\{a, b, c, d\}$.
The probability of node $a$ can influence node $d$ is $\mathsf{Pr}(a,d)=\mathsf{Pr}_{a}(\mathsf{P}_{1})+\mathsf{Pr}_{a}(\mathsf{P}_{2})+\mathsf{Pr}_{a}(\mathsf{P}_{3})  =  0.4+0.3*0.2+0.7*0.5*0.2=0.53$.

Definition~\ref{def:actProb} shows the contribution of node $v$ to marginal gain $mg(u | \mathsf{S})$ in \LT{} model.  

\begin{definition}[Node Contribution in \LT{}]\label{def:actProb}
The contribution of node $v$ to marginal gain $mg(u | \mathsf{S})$ is $\Phi(u,v)= \mathsf{RC}(u) \times\mathsf{Pr}(u,v)$.
$\Phi(u,v) = \mathsf{RC}(v)$ if $\mathsf{RC}(u) \times \mathsf{Pr}(u,v) \geq \mathsf{RC}(v)$.
\end{definition}

\stitle{Example} In Figure~\ref{fig:ltrsd}(b), node $c$'s contribution to marginal gain $mg(b|\{a\})$ is {$\Phi(b,c) =  \mathsf{RC}(b) \times\mathsf{Pr}(b,c) = 0.3*0.5=0.15$}.

Formally, given a residual social network \rG$(V,E,W,C)$ and seed set \Set{}, 
the marginal gain of node $u$ consists of i) the residual capacity of node $u$, and
ii) node contributions from other influenced nodes. Specifically, 
$$ mg(u | \mathsf{S}) = \mathsf{RC}(u) + \sum_{v \in V - (\mathsf{S} \cup \{u\})} \Phi(u,v).$$


\stitle{Example} In Figure~\ref{fig:ltrsd}(b), the marginal gain $mg(b | \{a\}) = \mathsf{RC}(b) + \Phi(b,c) + \Phi(b,d) = 0.3 + 0.3*0.5+ 0.3*0.5*0.2 = 0.48$.

\subsubsection{\RCELF{} Marginal Gain Computation in \IC{} Model}   
Comparing to the contribution of node $v$ to $mg(u | \mathsf{S})$ in \LT{} model, it is more complex in \IC{} model.
The reason is that the active node $u$ will definitely influence its inactive neighbor $v$ in \LT{} model, i.e., $\theta_{v} = \theta_{v} - w(u,v)$.
However, the active node $u$ may not influence its inactive neighbor $v$ as the influence process is a random coin-flip process with bias $w(u,v)$ in \IC{} model.

To illustration, consider the probability that node $c$ can be activated by node $a$ in Figure~\ref{fig:ltrsd}(a) in both \LT{} and \IC{} model.
There are two paths from $a$ to $c$, i.e., $\mathsf{P}_{1}=\{a,c\}$ and $\mathsf{P}_{2}=\{a,b,c\}$ respectively.
In \LT{} model, the influence probability $\Phi(a,c) = \mathsf{Pr}_{a}(\mathsf{P}_1) + \mathsf{Pr}_{a}(\mathsf{P}_2) = 0.3+0.35=0.65$ by Definition~\ref{def:ltpathprob}.
I.e, the node $c$ will be activated if $\theta_{c} \leq 0.65$.
In \IC{} model, however, node $a$ influence node $c$ with probability $\mathsf{Pr}_{1}=w(a,c)=0.3$ via path $\mathsf{P}_1$ and with probability $\mathsf{Pr}_{2}= w(a,b) \times w(b,c)=0.35$ via path $\mathsf{P}_2$.
Thus, the total activated probability of node $c$ by node $a$ is $1-(1-\mathsf{Pr}_{1})(1-\mathsf{Pr}_{2}) = 1-0.7*0.65=0.545$ according to conditional probability theory.

To facilitate the discussion of node marginal gain contribution in \IC{} model,
we divide the reachable nodes of $u$ into two groups: (i) \emph{shared-nothing} set $\mathsf{SN}$ and (ii) \emph{shared-edge} set $\mathsf{SE}$.
For example, Figure~\ref{fig:ltrsd}(a), node $a$'s reachable node $c$ is \emph{shared-nothing} node as the paths from $a$ to $c$ does not share any edge, i.e., $\{a,c\}$ and $\{a,b,c\}$.
$b$ also is $a$'s \emph{shared-nothing} nodes.
However, $d$ is a \emph{shared-edge} node as paths $\{a,b,c,d\}$ and $\{a,c,d\}$ shared a common edge $e(c,d)$.
The node contribution of each node $v$ in \emph{shard-nothing} node set $\mathsf{SN}$ to the marginal gain $mg(u|\mathsf{S})$ as follows.
\begin{definition}[Influence Probability in Single Path]\label{def:pathprob}
Given a path $\mathsf{P}=\{v_1, v_2,$ $\cdots, v_m \}$ in \rG$(V,E,W,C)$, seed set $\mathsf{S}$,
The probability of node $v_1$ influences $v_m$ is $\mathsf{Pr}_{v_1}(\mathsf{P})=  \prod_{i=1}^{m-1}(\mathsf{RC}(v_i)w(v_i, v_{i+1}) )$ in \IC{} model.
\end{definition}

Suppose there is a set of paths $\mathbb{P}$ in which node $u$ can influence node $v \in \mathsf{SN}$ (i.e., $v$ in \emph{shared-nothing} node set).
According to conditional probability theory, the probability node $u$ influence node $v$ is
$$\mathsf{Pr}(u,v) = 1-\prod_{\mathsf{p} \in \mathbb{P}} (1 - \mathsf{Pr}_{u}(\mathsf{P})).$$
The total contribution of node $v$ to $mg(u | \mathsf{S})$ in \IC{} model is $\Phi(u,v)= \mathsf{RC}(v) \times \mathsf{Pr}(u,v)$.

\stitle{Example} In Figure~\ref{fig:ltrsd}(a), there are two paths from node $a$ to node $c$: $\mathsf{P}_1=\{a,b,c\}$ and $\mathsf{P}_2=\{a,c\}$.
The contribution of $c$ to $mg(a | \emptyset)$ is equivalent to $\mathsf{RC}(a) \times \mathsf{Pr}(a,c) = 1.0 * (1-(1-0.3)*(1-0.7*0.5)) = 0.545$.

For the nodes in \emph{shared-edge} set $\mathsf{SE}$, it is quite difficult to analyze the active probability from an active node $u$ at step $t$.
Fortunately, inspired by the \SG{} approach, the active probability of these nodes in \emph{shared-edge} set could be obtained by running $r$ times Monte-Carlo simulations.
We then define the marginal gain contribution of every node in \emph{shared-edge} set in Definition~\ref{def:secontribute}.

\begin{definition}[\emph{Shared-edge} Node Contribution]\label{def:secontribute}
Given residual-based graph \rG$(V,E,W,C)$ and seed set $\mathsf{S}$.
For each node $v \in \mathsf{SE}$, the influence probability $\Phi(u,v)$ is obtained by Monte-Carlo simulation.
The node contribution of node $v$ to $mg(u | \mathsf{S})$ in \IC{} model is $\Phi(u,v)= \mathsf{RC}(v) \times (\mathsf{RC}(u) \mathsf{Pr}(u,v)) $.
\end{definition}

Finally, the marginal gain $mg(u | \mathsf{S})$ in \IC{} model can be computed by 
$mg(u | \mathsf{S}) = \mathsf{RC}(u) + \sum_{v \in V - (\mathsf{S} \cup \{u\})} \Phi(u,v)$.

\subsubsection{Updating Node Residual Capacity}
During \RCELF{} seed node selection procedure, suppose node $u$ is selected as seed node at step $t$ (i.e., $u = \argmax_{v \in V-\mathsf{S}} mg(v |\mathsf{S})$),
the residual capacity of every $u$'s reachable node $v$ will be updated by $\mathsf{RC}(v) = \mathsf{RC}(v) - \Phi(u,v)$ accordingly, it will be ignored if $\mathsf{RC}(v) \leq 0$ in subsequent seed selections.

\subsection{Implementation and Optimizations}\label{sec:rcelfopt}
In this section, we present the sketch of \RCELF{} approach with two performance optimization techniques.

\stitle{\RCELF{} Approach} The sketch of our residual-based approaches \RCELF{} for \IMP{} with \LT{} and \IC{} model as follows:
\begin{enumerate}[leftmargin=*]
  \item Init influence max-heap $\mathcal{H}$, it builds a max-heap by using the marginal gain as key value.
  \item Identify seed node $u$ (i.e., $\argmax_{u \in V-\mathsf{S}} mg(u | \mathsf{S})$) by $\mathcal{H}$ efficiently, insert it into seed set $\mathsf{S} \leftarrow \mathsf{S} \cup \{u\}$.
  \item Update residual capacity of each node $v$ in \rG{}, i.e., $\forall v \in V - \mathsf{S}$, $\mathsf{RC}(v) \leftarrow \mathsf{RC}(v) - \Phi(u,v)$. Node $v$ will be discarded in \rG{} if $\mathsf{RC}(v) \leq 0$.
  \item Repeat Step (2) and (3), until $k$ seed nodes are selected.
\end{enumerate}

In the subsequent section, we improve the performance of \RCELF{} by (1) proposing an efficient marginal gain computation algorithm and (2) reducing max-heap update cost.

\begin{algorithm}
    \small
    \caption{\textsf{MCSMG}($\mathsf{rG}(V,E,W,C)$, $u$, $\mathsf{S}$, $r$)} \label{alg:MConRG}
    \begin{algorithmic}[1]     \label{alg:rmc_init}
    \State Initialize map $\Phi \leftarrow \emptyset$ \Comment{reachable times from $u$ to $v$}
    \For{$i$ from $1$ to $r$} \Comment{$r$ Monte-Carlo simulations}
        \State Queue $q.\textsf{enqueue}(u)$
        \While{!IsEmpty($q$)}
            \State node $\textsf{tmp} \leftarrow q.\textsf{dequeue}()$
       		\For{each $v$ in \outu{$\textsf{tmp}$} and it is inactive}
       			\If{$\textsf{RAND}() < \mathsf{RC}(\textsf{tmp}) \cdot w(\textsf{tmp},v)$}\label{alg:rmc_act}
       				\State $q.\textbf{enqueue}(v)$ \Comment{set $v$ as active}
       				\State $\Phi(u,v) \leftarrow \Phi(u,v) + 1$  \label{alg:phiuv}
       			\EndIf
       		\EndFor
        \EndWhile
    \EndFor
    \State $mg(u | \mathsf{S}) \leftarrow  \mathsf{RC}(u)$
    \For {each $\Phi(u,v)$ in $\Phi$}  \label{ln:phis}
        \State $\Phi(u,v) \leftarrow \Phi(u,v) / r *\mathsf{RC}(u)$
        \State $mg(u | \mathsf{S}) \leftarrow mg(u | \mathsf{S}) + \Phi(u,v)$
    \EndFor \label{ln:phie}
    \State Return $mg(u | \mathsf{S})$ and $\Phi$
    \end{algorithmic}
\end{algorithm}

\stitle{Efficient Marginal Gain Computation}
Intuitively, we compute the marginal gain of each node, i.e., for node $u$ and seed set \Set{},  $mg(u | \mathsf{S})$ is initialized to $\mathsf{RC}(u)$.
We enumerate all the paths from $u$ to each node $v \in V - \mathsf{S}$ and calculate $\Phi(u,v)$, then accumulate $\Phi(u,v)$ to $mg(u | \mathsf{S})$ in \LT{} model.
However, the computation cost is exponential to the number of edges in \rG{}. Hence, it is impractical in median or large social networks.
In addition, the contribution of \emph{shared-edge} nodes cannot compute exactly as \emph{shared-nothing} nodes in \IC{} model.
To address the above issues, we devise a Monte-Carlo simulation based marginal gain computation algorithm (cf. Algorithm~\ref{alg:MConRG}).
The main idea of Algorithm~\ref{alg:MConRG} is that it incorporates all $\Phi(u,v)$ computations in one batch Monte-Carlo simulation process.
Algorithm~\ref{alg:MConRG} shows the exact steps about $mg(u|\mathsf{S})$ computation in \IC{} model.
In each Monte-Carlo simulation, it takes the residual capacity of each inactive node $v$ into consideration by flipping coins with probability
$\mathsf{RC}(c) \cdot w(u,v)$ (cf. Line~\ref{alg:rmc_act}) instead of only $w(u,v)$.
$\Phi(u,v)$ counts the number of activated times of node $v$ among $r$ times simulation (cf. Line~\ref{alg:phiuv}).
Finally, for each node $v$, its $\Phi(u,v)$ compute as Definition~\ref{def:actProb} from Line~\ref{ln:phis} to Line~\ref{ln:phie}.
It is worthing to note Algorithm~\ref{alg:MConRG} is applicable to \LT{} models.
For example, we only need revise the node activation manner (cf. Line~\ref{alg:rmc_act}) for \LT{} model.

\stitle{Heap Updates Optimization}
\RCELF{} identifies the seed node $u$ (i.e., $\argmax_{u \in V-\mathsf{S}} mg(u | \mathsf{S})$) with max-heap $\mathcal{H}$.
The marginal gain of nodes in max-heap $\mathcal{H}$ need recompute as it is out-of-date after each seed node selection iteration.
i.e., the current marginal gain of node $u$ is computed with an out-of-date seed set, denoted by $\mathsf{S}_o$.
However, it should be $mg(u | \mathsf{S})$, where $S$ is latest seed set.
Thus, the performance of \RCELF{} approach is very sensitive to the number of node marginal gain computations in Step (2) to identify the next seed node $\argmax_{u \in V-\mathsf{S}} mg(u | \mathsf{S})$.
Here, we propose an upper bound for the marginal gain of node $u$ with latest seed set $\mathsf{S}$, denote by $\overline{mg}(u|\mathsf{S})$.
It reduces the number of marginal gain computations significantly.

\begin{lemma}[Upper Bound of $mg(u|\mathsf{S})$]\label{lem:sigma}
For each node $v \in V$, its residual capacity and marginal gain are $\mathsf{RC}_{o}(v)$ and $mg(v| \mathsf{S}_o)$ when the seed set is $\mathsf{S}_o$ at step $t-1$.
The seed set is $\mathsf{S}$ at step $t$ (i.e., $\mathsf{S}_o \subset \mathsf{S}$), the upper bound $\overline{mg}(u | \mathsf{S})$ is $\mathsf{RC}(u) + \mathsf{RC}(u) \sum_{v \in \mathsf{Out}(u)} w(u,v) * \mathsf{RC}_{o}(v) * mg(v| \mathsf{S}_o)$.

\begin{proof}
For each node $v \in V$, $\mathsf{RC}_{o}(v) \geq \mathsf{RC}(v)$ as the residual capacity is diminishing during seed node selection process.
With $mg(v| \mathsf{S}_o) \geq mg(v| \mathsf{S})$ where $\mathsf{S}_o \subset \mathsf{S}$ (by submodularity), we have:
\begin{small}
\begin{align*}
  \bar{mg}(u | \mathsf{S}) & = \mathsf{RC}(u) + \mathsf{RC}(u)\sum_{v \in \mathsf{Out}(u)} w(u,v) * \mathsf{RC}_{o}(v) * mg(v| \mathsf{S}_o)  \\
  & \geq \mathsf{RC}(u) + \mathsf{RC}(u)\sum_{v \in \mathsf{Out}(u)} w(u,v) * \mathsf{RC}(v) * mg(v| \mathsf{S})  \\
  & \geq \mathsf{RC}(u) + \sum_{v \in V - (\mathsf{S} \cup u)} \Phi(u,v)\\
  & = mg(v| \mathsf{S}) 
\end{align*}
\end{small}
Thus, we have $\overline{mg}(u | \mathsf{S})\geq mg(u | \mathsf{S})$.
\end{proof}
\end{lemma}


Consider $mg(u | \mathsf{S}_o)$ with seed set $\mathsf{S}_o$, the marginal gain upper bound $\overline{mg}(u | \mathsf{S})$ will be computed with constant cost at first,
then the max-heap $\mathcal{H}$ is updated with $\overline{mg}(u | \mathsf{S})$.
Algorithm~\ref{alg:MConRG} will be incurred to compute the exact marginal gain $mg(u|\mathsf{S})$ if and only if the root is $\overline{mg}(u | \mathsf{S})$.
Inherently, Lemma~\ref{lem:sigma} works as a filter which reduces lots of expensive exact marginal gain computations.


\subsection{\RCELF{} Analysis}\label{sec:ana}
We analyze the property of influence spread function $\delta(\cdot)$ in \RCELF{} at Theorem~\ref{the:delta}, then prove the result accuracy guarantee of \RCELF{} in Lemma~\ref{lem:ratio}.

\begin{theorem}~\label{the:delta}
The influence spread function $\delta(\cdot)$ in \RCELF{} with \IC{} model is  (i) \emph{non-negative}, (ii) \emph{monotone}, and (iii) \emph{submodular}.
\end{theorem}

\begin{proof}
Suppose the selected seed nodes  from $1$st iteration to $k$th iteration are $v_1$, $v_2$, $\cdots$, $v_k$.
The corresponding seed sets are $\mathsf{S}_1, \mathsf{S}_2, \cdots, \mathsf{S}_k$, and $\mathsf{S}_0 = \emptyset$.
Thus, $\delta(\mathsf{S}_k) = mg(v_k | \mathsf{S}_{k-1}) + \delta(\mathsf{S}_{k-1}) = \sum_{i=1}^{i=k} mg(v_i | \mathsf{S}_{i-1})$.
$\forall i \in  [1, k], mg(v_i | \mathsf{S}_{i-1}) \geq 0$ in \RCELF{} approach.
Then, $\delta(\mathsf{S}) \geq 0$ and $\forall~\mathsf{S} \subseteq \mathsf{S}'$, $\delta(\mathsf{S}) \leq \delta(\mathsf{S}')$.
Hence, $\delta(\cdot)$ in \RCELF{} is (i) \emph{non-negative} and (ii) \emph{monotone}.

Since the residual capacity of each node $v$ will be diminished, cf. Step (3), during seed node selection process,
$\forall~ \mathsf{S} \subset \mathsf{S}'$ we have
\begin{align*}
& mg(u | \mathsf{S}) \geq mg(u | \mathsf{S}')  \\
\Rightarrow ~& \delta(\mathsf{S}) + mg(u | \mathsf{S}) - \delta (\mathsf{S}) \geq  \delta(\mathsf{S}') + mg(u | \mathsf{S}') - \delta (\mathsf{S}') \\
\Rightarrow  ~&\delta(\mathsf{S} \cup \{ u \}) -\delta (\mathsf{S}) \geq \delta(\mathsf{S}' \cup \{ u \}) - \delta(\mathsf{S}')
\end{align*}
Thus, $\delta(\cdot)$ in \RCELF{} is submodular, the proof complete.
\end{proof}

We then show the result accuracy guarantee of \RCELF{} in Lemma~\ref{lem:ratio} as follows.

\begin{lemma}~\label{lem:ratio}
	Let size-$k$ set $\mathsf{S}^*$ be optimal set of \IMP{}, i.e., $\delta(\mathsf{S}^*)$ has maximal value of all $k$-element sets.
	\RCELF{} returns size-$k$ set $\mathsf{S}$, which guarantees $$\delta(\mathsf{S}) \geq (1-1/e)*\delta(\mathsf{S}^*).$$
\end{lemma}

\begin{proof}
Since the influence spread function $\delta(\cdot)$  in \RCELF{} is non-negative, monotone and submodular (cf. Theorem~\ref{the:delta}),
it guarantees that $\mathsf{S}$, returned from \RCELF{}, provides a $(1 - {1}/{e})$-approximation ratio, as proved in~\cite{nemhauser1978analysis,kempe2003maximizing}.
\end{proof}

Inherently, \RCELF{} works similar with \SG{} and \PMC{}.
For example,  \RCELF{} is built upon the residual-based graph.
The residual capacity diminishing process is similar to removing nodes and edges in generated snapshots in \SG{} and \PMC{}.
Inspired from~\cite{ohsaka2014fast}, we analyze the number of necessary Monte-Carlo simulations in \RCELF{} to guarantee a $(1-1/e-\epsilon)$-approximation result as follows.

\begin{lemma}~\label{lem:rtimes}
By setting Monte-Carlo simulation times $r=O(\frac{n^2\log n\log \binom{n}{k}}{\epsilon^2})$, \RCELF{} achieves a $(1-1/e-\epsilon)$-approximation ratio with probability $1-1/n$.
\end{lemma}
\begin{proof}
Let $\mathcal{S}$ be the set of every possible result set $\mathsf{S}$,  $|\mathcal{S}| = \binom{n}{k}$.
For any result set instance $\mathsf{S} \in \mathcal{S}$, $\delta_{i} (\mathsf{S})$ is $\mathsf{S}$'s influence spread value of the $i$-th Monte-Carlo simulation at \G{}, $\delta_{i}(\mathsf{S}) \in [1, n]$.
Let $\bar{\delta}(\mathsf{S}) = \frac{1}{r} \sum_{i=1}^{r}\delta_{i}(\mathsf{S})$.
By applying Hoeffding's ineqaulity (Theorem 3 in~\cite{ohsaka2014fast}),
for any $\mathsf{S} \in \mathcal{S}$, $\frac{1}{n}|\bar{\delta}(\mathsf{S}) - \delta(\mathsf{S})| \leq \epsilon $ holds with at least probability $1-2 e^{-2 r \epsilon^2} \binom{n}{k}$ .
By choosing $r = O(\frac{n^2\log n\log \binom{n}{k}}{\epsilon^2})$,  the above conclusion holds with at least probability $(1 - 1/n)$ .
By applying Lemma~\ref{lem:ratio}, we have $\bar{\delta}(\mathsf{S}) \geq (1-1/e-\epsilon)*\delta(\mathsf{S}^*)$ with $r$ Monte-Carlo simulation times.
\end{proof}





\section{Experimental Evaluation}\label{sec:exp}

In this section we evaluate \RCELF{} and present our empirical findings.
In Section~\ref{sec:setting} we describe the experimental setting.
In Section~\ref{sec:evaluation} we compare our proposal with existing approximate approaches in real datasets, and investigate the effectiveness of the optimization technique.

\begin{figure*}
  \centering
  \includegraphics[width=1.0\columnwidth]{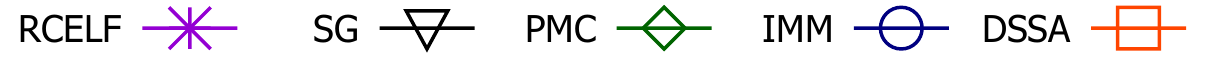}
  \begin{tabular}{cccc}
    \includegraphics[width=0.45\columnwidth]{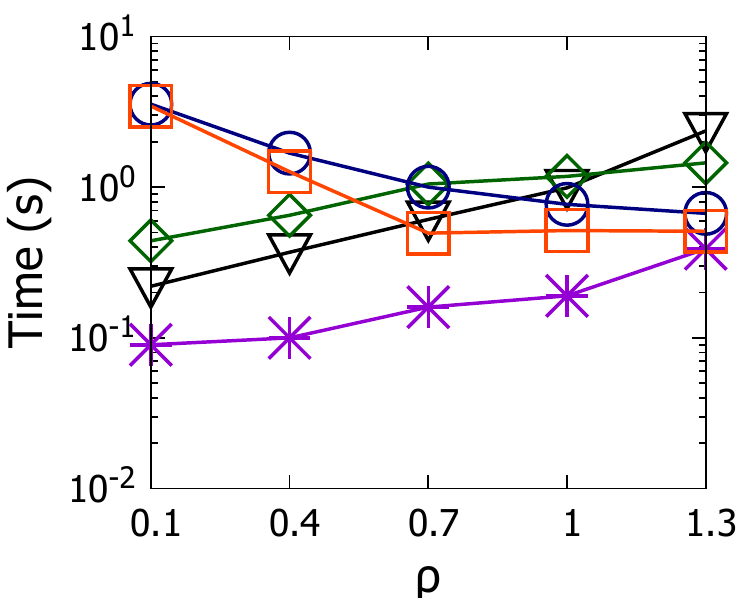}
    &
    \includegraphics[width=0.45\columnwidth]{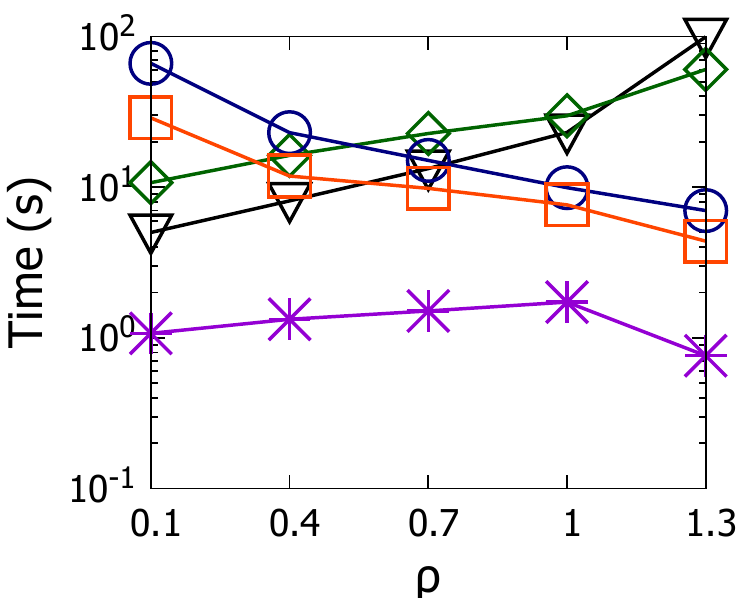}
    &
    \includegraphics[width=0.45\columnwidth]{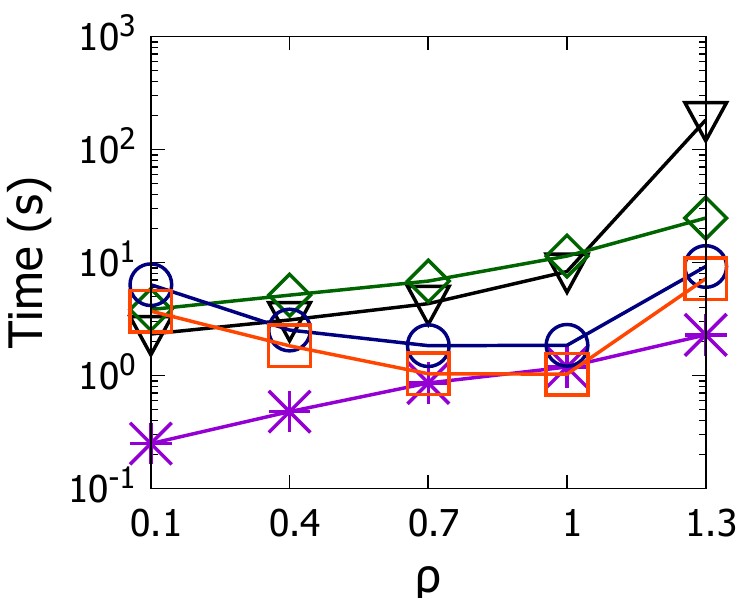}
    &
    \includegraphics[width=0.45\columnwidth]{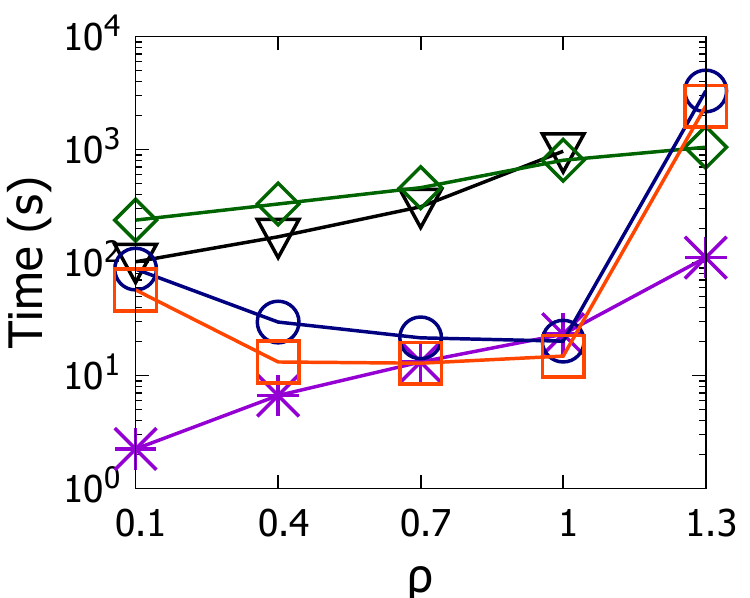}
    \\
    (a) \nethept{} & (b) \dblp{} & (c) \twitter{} & (d) \livejournal{}
  \end{tabular}
  \vspace{-2mm}
  \caption{Execution time vs. $\rho$ with $k=5$, \WC{}}\label{fig:k=5varyp_time_wc}
  
  \begin{tabular}{cccc}
    \includegraphics[width=0.45\columnwidth]{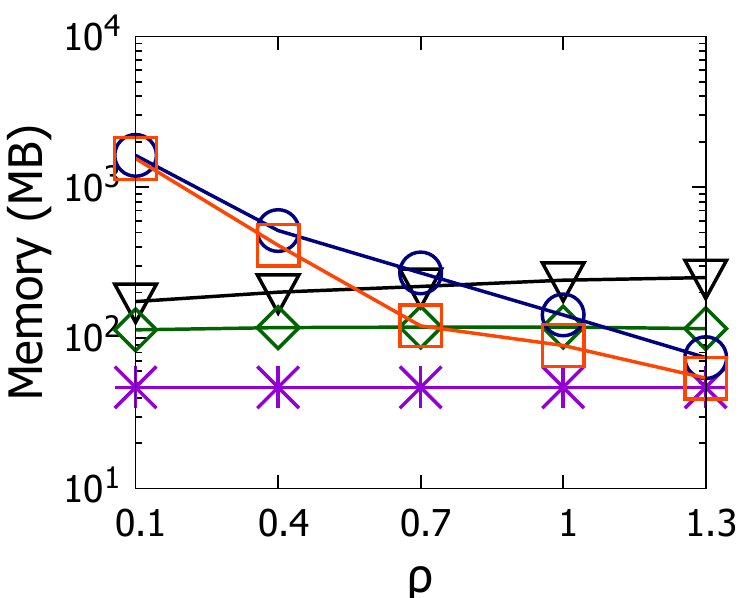}
    &
    \includegraphics[width=0.45\columnwidth]{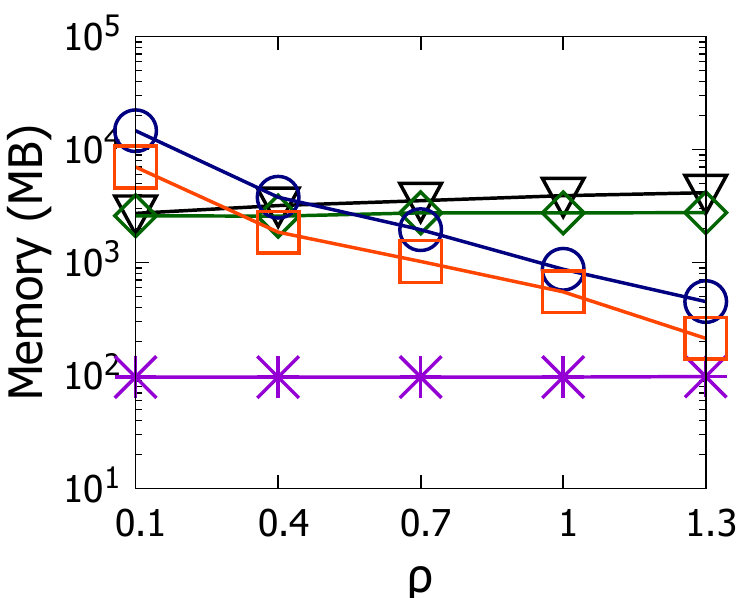}
    &
    \includegraphics[width=0.45\columnwidth]{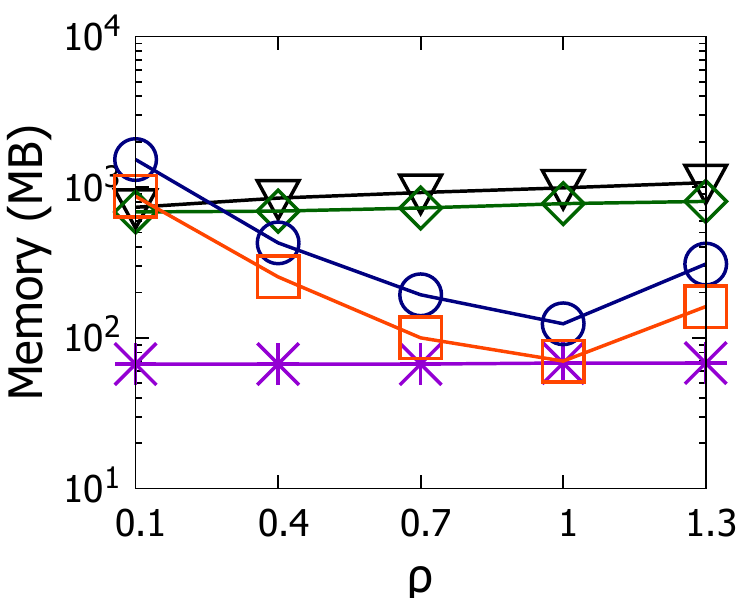}
    &
    \includegraphics[width=0.45\columnwidth]{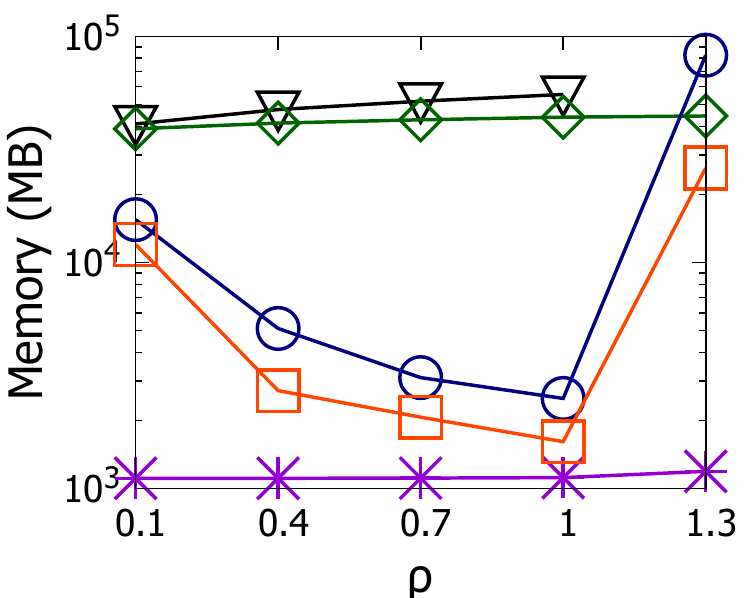}
    \\
    (a) \nethept{} & (b) \dblp{} & (c) \twitter{} & (d) \livejournal{}
  \end{tabular}
  \vspace{-2mm}
  \caption{Memory consumption vs. $\rho$ with $k=5$, \WC{}}\label{fig:k=5varyp_memory_wc}
  
  \includegraphics[width=1.0\columnwidth]{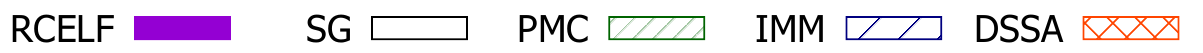}
   \begin{tabular}{cccc}
    \includegraphics[width=0.45\columnwidth]{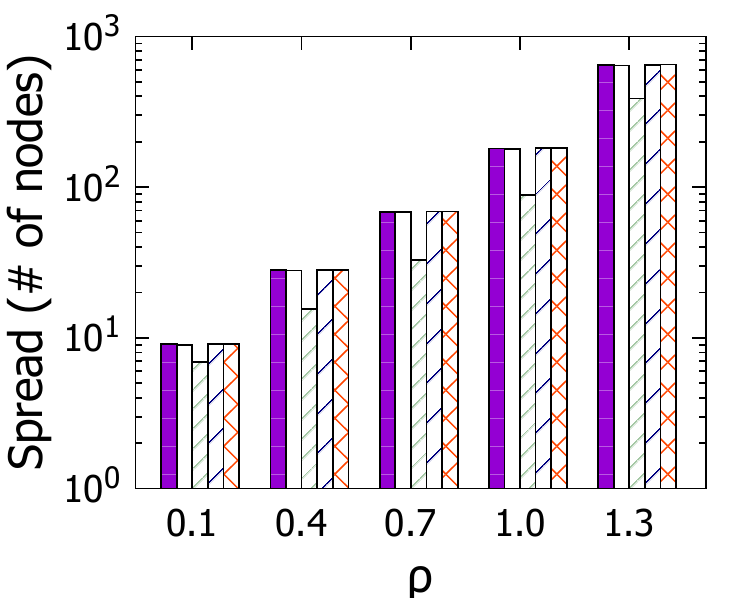}
    &
    \includegraphics[width=0.45\columnwidth]{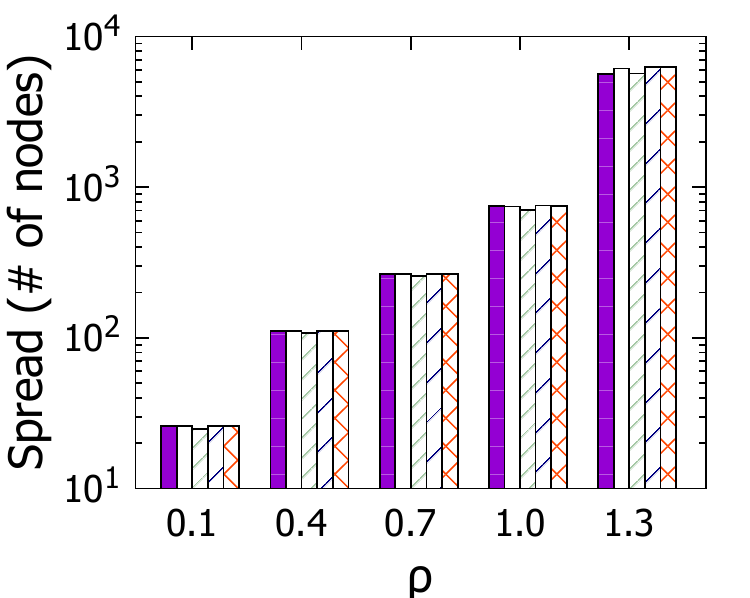}
    &
    \includegraphics[width=0.45\columnwidth]{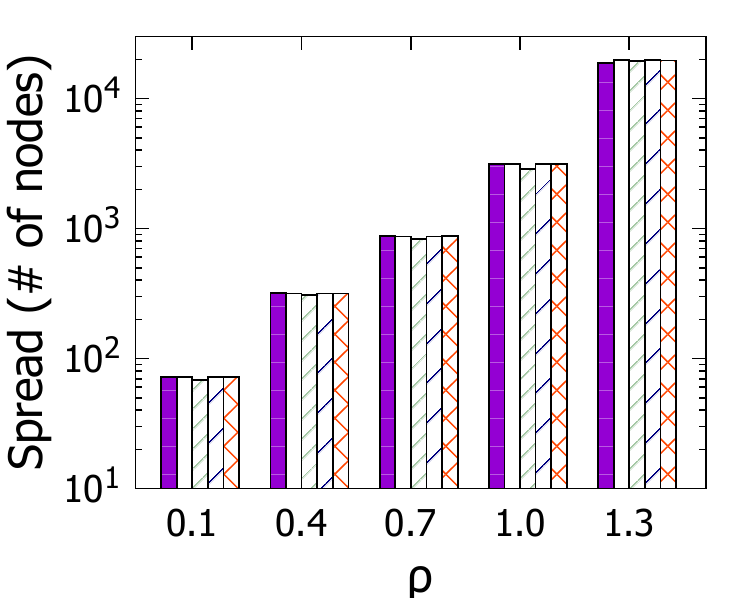}
    &
    \includegraphics[width=0.45\columnwidth]{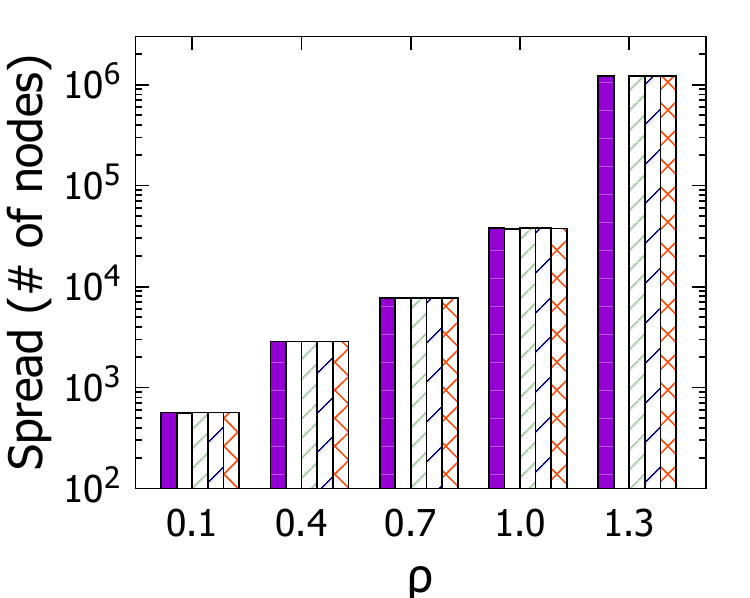}
    \\
    (a) \nethept{} & (b) \dblp{} & (c) \twitter{} & (d) \livejournal{}
  \end{tabular}
  \vspace{-2mm}
  \caption{Result quality vs. $\rho$ with $k=5$, \WC{}}\label{fig:k=5varyp_quality_wc}
\end{figure*}

\vspace{-4mm}

\begin{figure*}
  \centering
  \includegraphics[width=1.0\columnwidth]{fig/key}
  \begin{tabular}{cccc}
    \includegraphics[width=0.45\columnwidth]{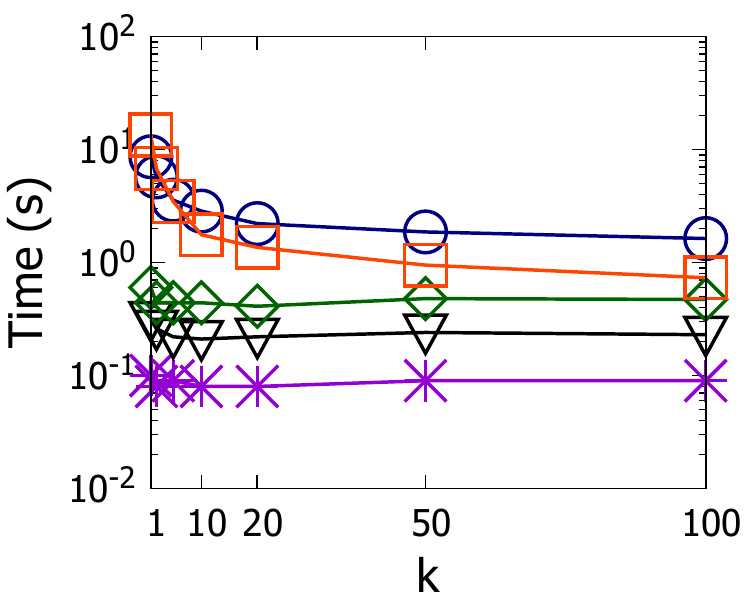}
    &
    \includegraphics[width=0.45\columnwidth]{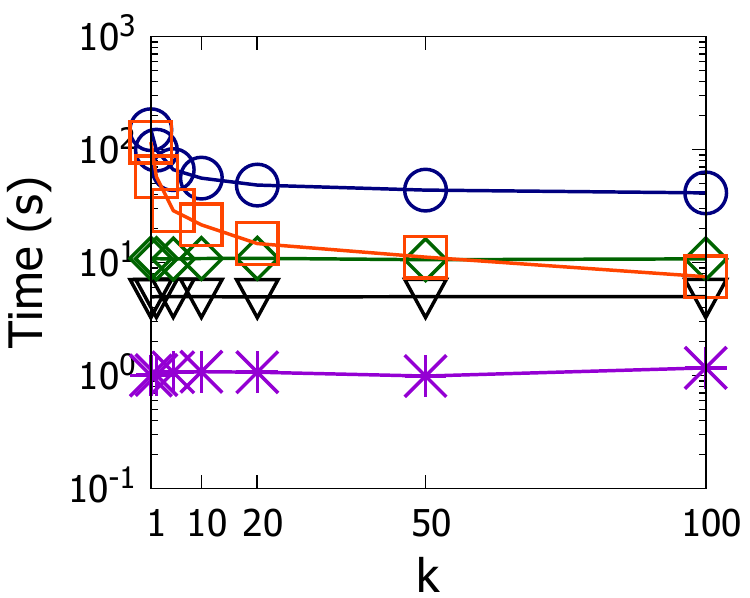}
    &
    \includegraphics[width=0.45\columnwidth]{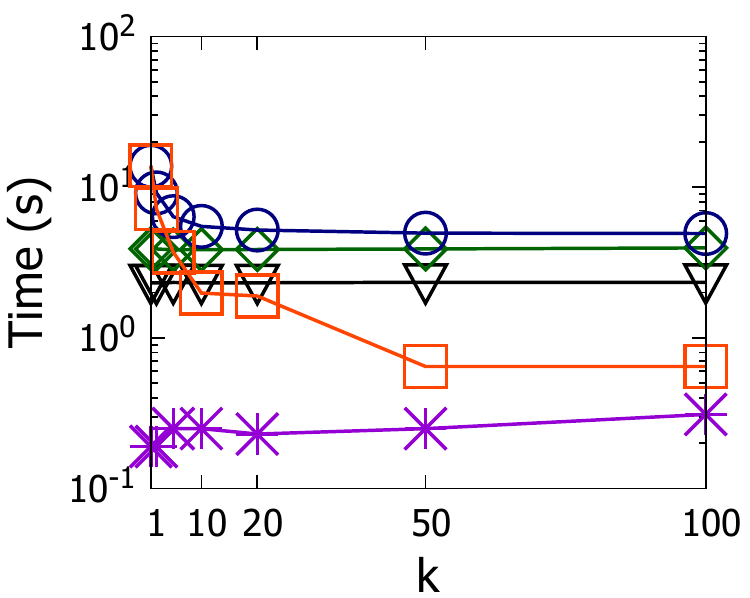}
    &
    \includegraphics[width=0.45\columnwidth]{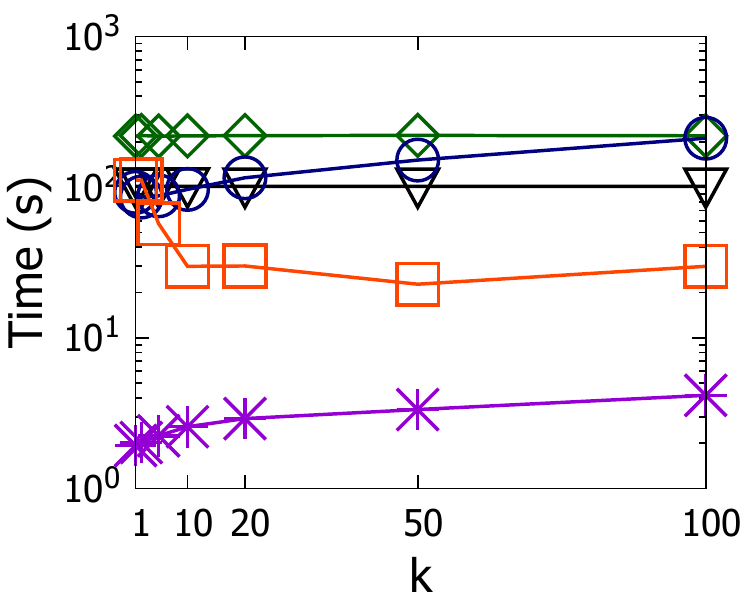}
    \\
    (a) \nethept{} & (b) \dblp{} & (c) \twitter{} & (d) \livejournal{}
  \end{tabular}
  \caption{Execution time vs. $k$ with $\rho=0.1$, \WC{}}\label{fig:p=0.1varyk_time_wc}
  
  \begin{tabular}{cccc}
    \includegraphics[width=0.45\columnwidth]{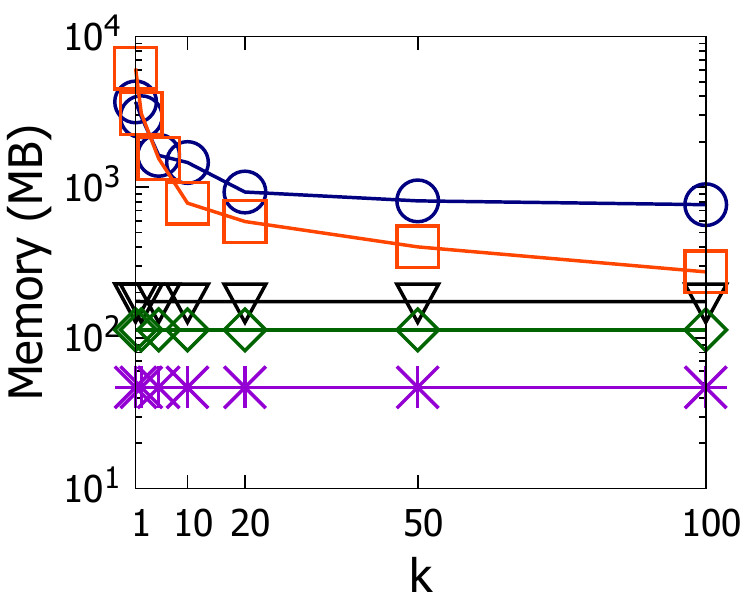}
    &
    \includegraphics[width=0.45\columnwidth]{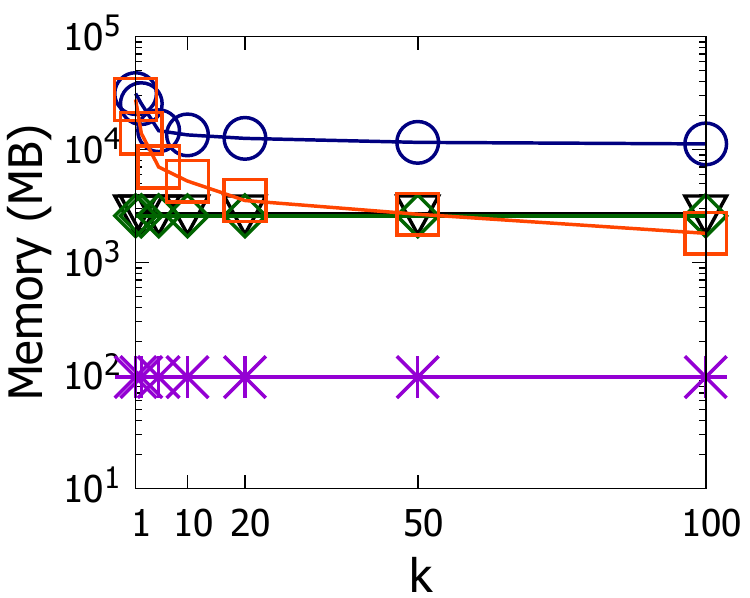}
    &
    \includegraphics[width=0.45\columnwidth]{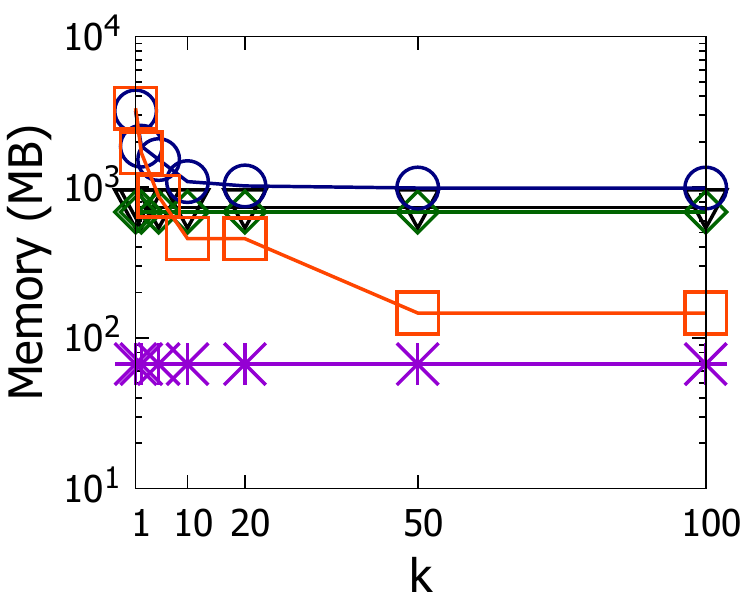}
    &
    \includegraphics[width=0.45\columnwidth]{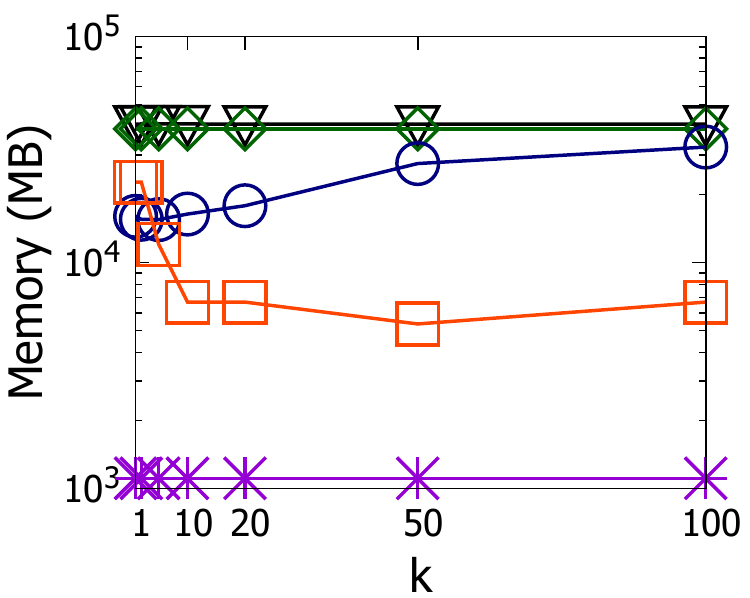}
    \\
    (a) \nethept{} & (b) \dblp{} & (c) \twitter{} & (d) \livejournal{}
  \end{tabular}
  \caption{Memory consumption vs. $k$ with $\rho=0.1$, \WC{}}\label{fig:p=0.1varyk_memory_wc}
  
  \includegraphics[width=1.0\columnwidth]{fig/key_spread}
  \begin{tabular}{cccc}
    \includegraphics[width=0.45\columnwidth]{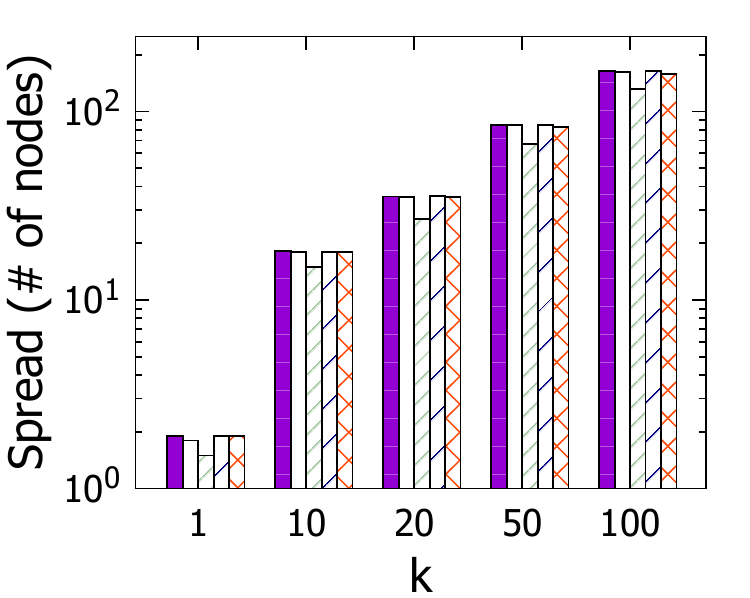}
    &
    \includegraphics[width=0.45\columnwidth]{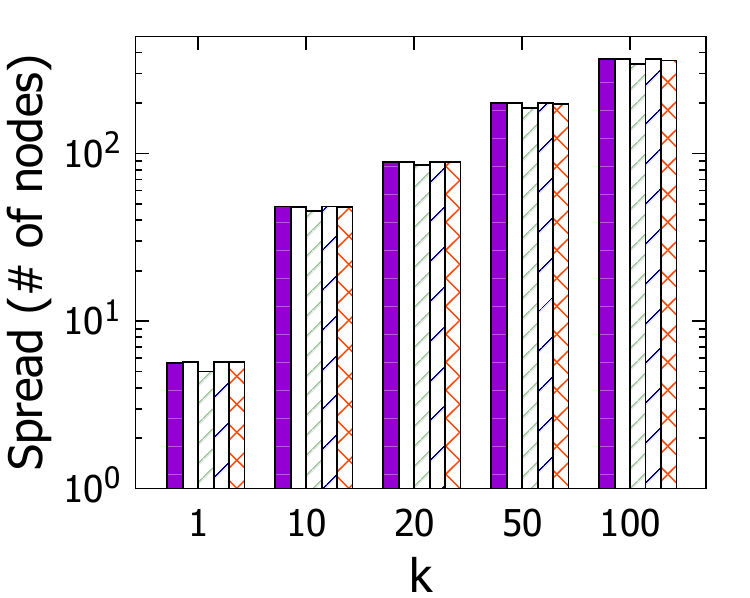}
    &
    \includegraphics[width=0.45\columnwidth]{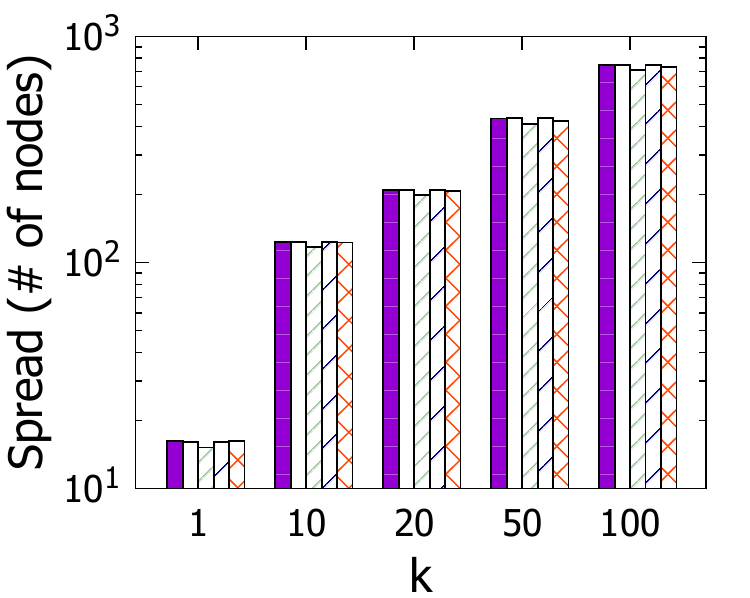}
    &
    \includegraphics[width=0.45\columnwidth]{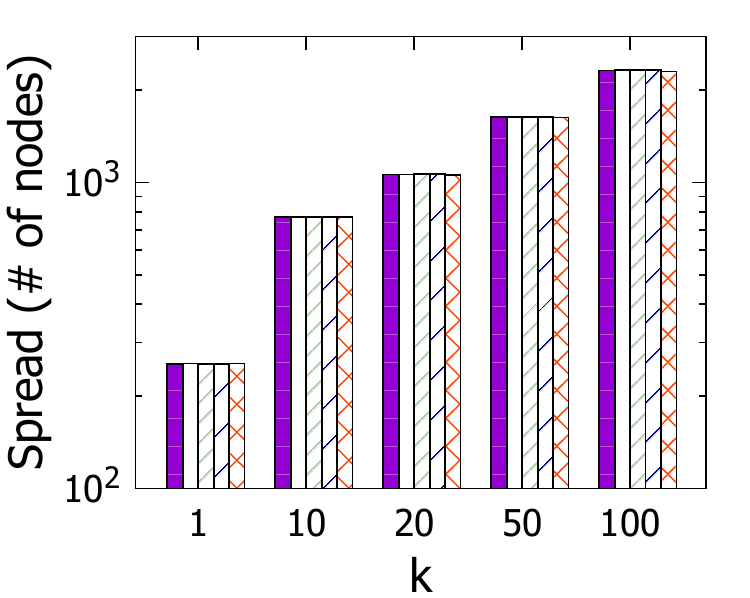}
    \\
    (a) \nethept{} & (b) \dblp{} & (c) \twitter{} & (d) \livejournal{}
  \end{tabular}
  \caption{Result quality vs. $k$ with $\rho=0.1$, \WC{}}\label{fig:p=0.1varyk_qualty_wc}
\end{figure*}

\subsection{Experimental Setting}\label{sec:setting}

\stitle{Datasets} We use five widely-used benchmarks for \IMP{}.
Table~\ref{tab:dataset} provides details on the number of nodes, edges, average degree, data type and sources of each dataset.

\begin{table}
	\centering
	\small
	\caption{Dataset characteristics}\label{tab:dataset}
	\resizebox{0.95\columnwidth}{!}{
	\begin{tabular}{|c|c|c|c|c|c|} \hline
		Datasets& $n$ & $m$ & Direct & Source \\ \hline \hline
		\nethept{} & 15K& 62K  & No & arxiv.org\\ \hline
		\twitter{} & 81K& 1.7M & Yes & snap.stanford.edu \\ \hline
		\dblp{} & 317K& 2M  &No & snap.stanford.edu \\ \hline
		\livejournal{} & 4.8M& 69M & Yes & snap.stanford.edu\\ \hline
        \twitterrv{} & 41.7M& 1.5G & Yes & an.kaist.ac.kr \\ \hline
	\end{tabular}
	}
\end{table}

\stitle{Algorithms} We compare our proposal \RCELF{} with four methods, namely, \SG{}~\cite{cheng2013staticgreedy}, \PMC{}~\cite{ohsaka2014fast}, \IMM{}~\cite{tang2015influence} and \DSSA{}~\cite{nguyen2016stop}.
We omit \CELF{} and \CELFplus{} as they are infeasible for datasets in Table~\ref{tab:dataset}.
The source code of \PMC{}, \IMM{} and \DSSA{} are from the homepage of the authors.
We implement \SG{} and \RCELF{} by C++.
All experiments run on Centos 7.4 with Intel Xeon E5-2620(2.1GHz) and 80GB memory.

\stitle{Diffusion Models}
We consider the widely used Weighted Cascade (\WC{}), Independent Cascade (\IC{}) and Linear Threshold (\LT{}) models.
In conventional \WC{} and \LT{} models, the weight of each edge $w(u,v)$ is set as $1/|$\inu{v}$|$.
We use $\rho/|\mathsf{In}{v}|$ in this work, where $\rho$ is a tunable parameter.
For \IC{} model, we follow the settings in~\cite{cheng2013staticgreedy}, e.g., the weight of each edge is $0.001$.

\stitle{Parameter Setting}
In order to guarantee the same result quality, we set $\epsilon =0.1$ in \IMM{}~\cite{huang2017revisiting}, \DSSA{}~\cite{nguyen2016stop} and our proposal \RCELF{}.
For \DSSA{} $\delta=1/|V|$,  $r=200$ for \SG{} and \PMC{}.
The performance metrics of \IMP{} are execution time, memory consumption.
Each plotted value corresponds to the average of measurements observed over 20 times.

\subsection{Performance Evaluation}\label{sec:evaluation}

\begin{figure*}
  \centering
  \includegraphics[width=1.0\columnwidth]{fig/key}
  \begin{tabular}{cccc}
    \includegraphics[width=0.45\columnwidth]{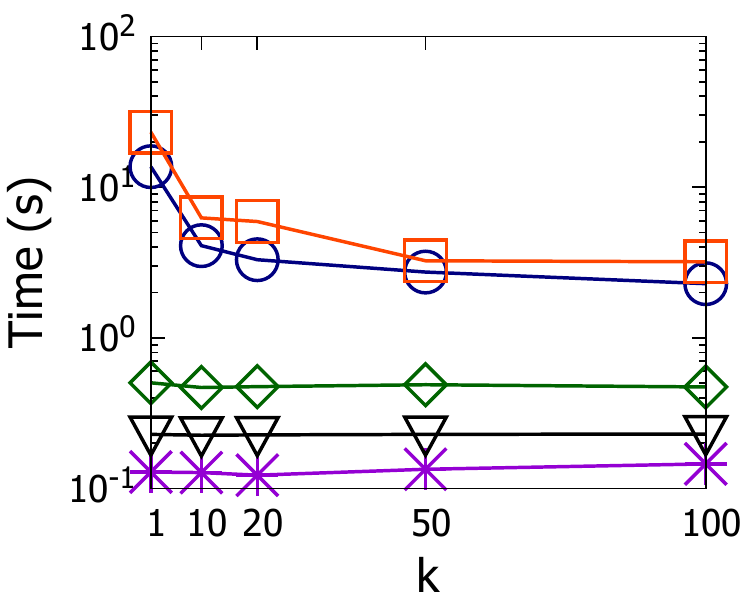}
    &
    \includegraphics[width=0.45\columnwidth]{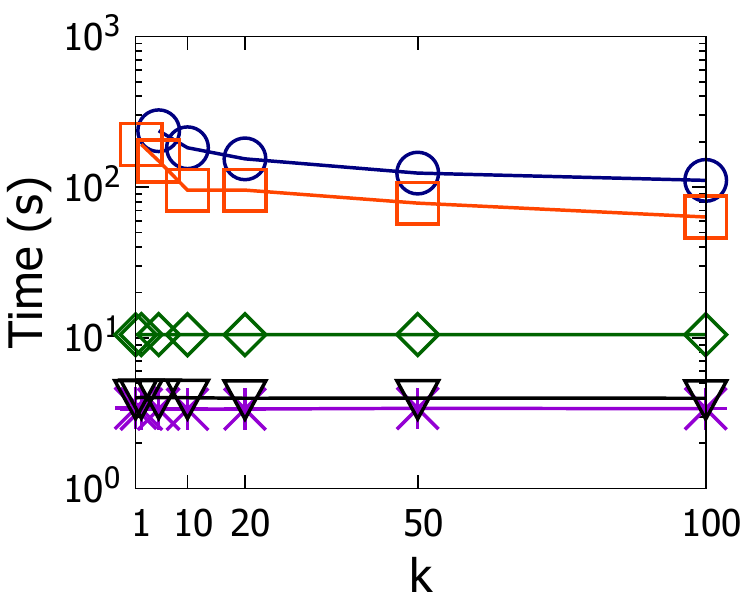}
    &
    \includegraphics[width=0.45\columnwidth]{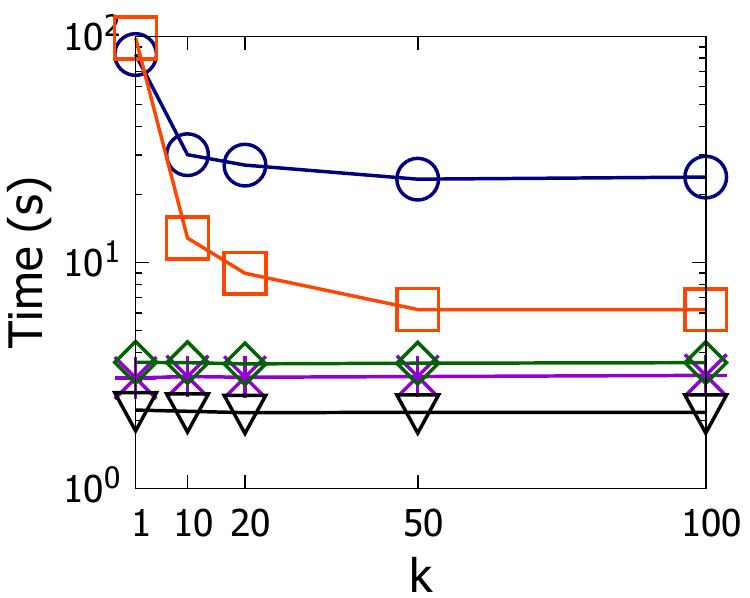}
    &
    \includegraphics[width=0.45\columnwidth]{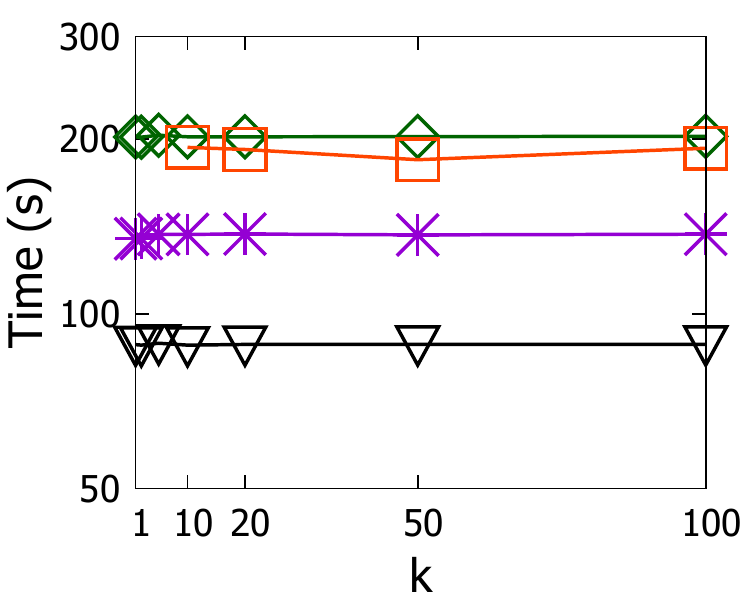}
    \\
    (a) \nethept{} & (b) \dblp{} & (c) \twitter{} & (d) \livejournal{}
  \end{tabular}
  \caption{Execution time vs. $k$ with {$w = 0.001$}, \IC{}}\label{fig:p=0.001varyk_time_ic}
  \begin{tabular}{cccc}
    \includegraphics[width=0.45\columnwidth]{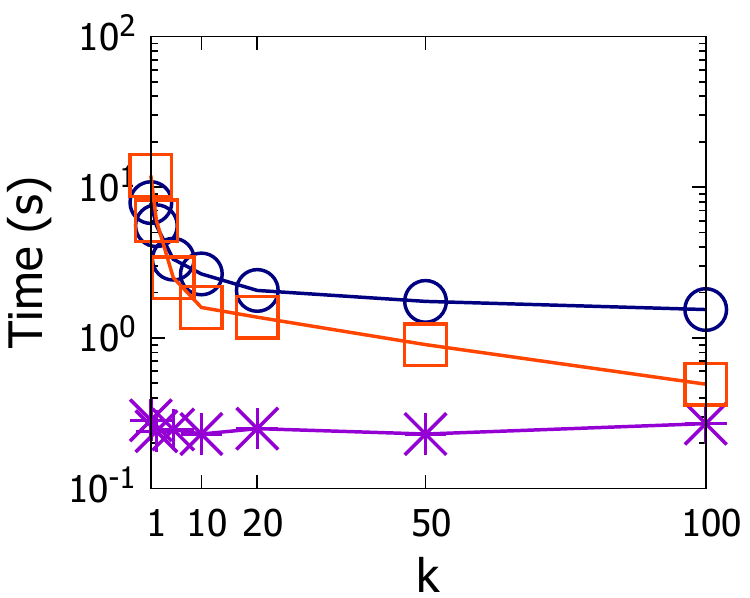}
    &
    \includegraphics[width=0.45\columnwidth]{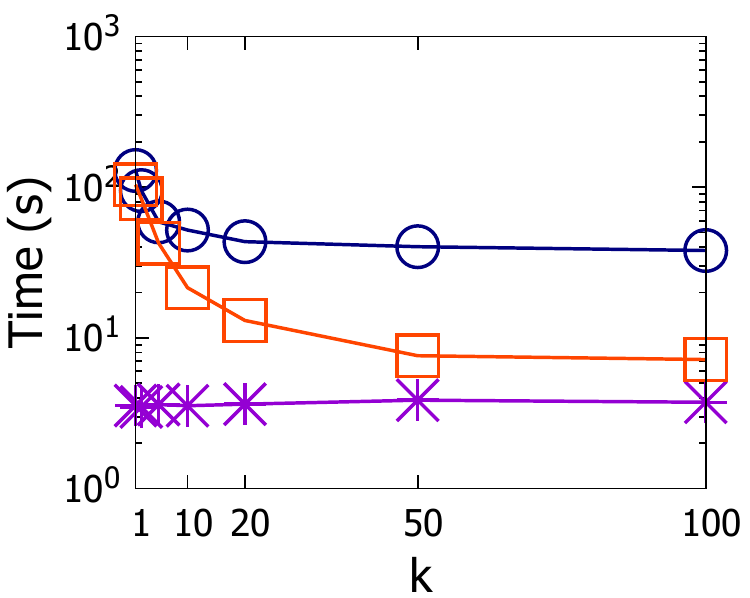}
    &
    \includegraphics[width=0.45\columnwidth]{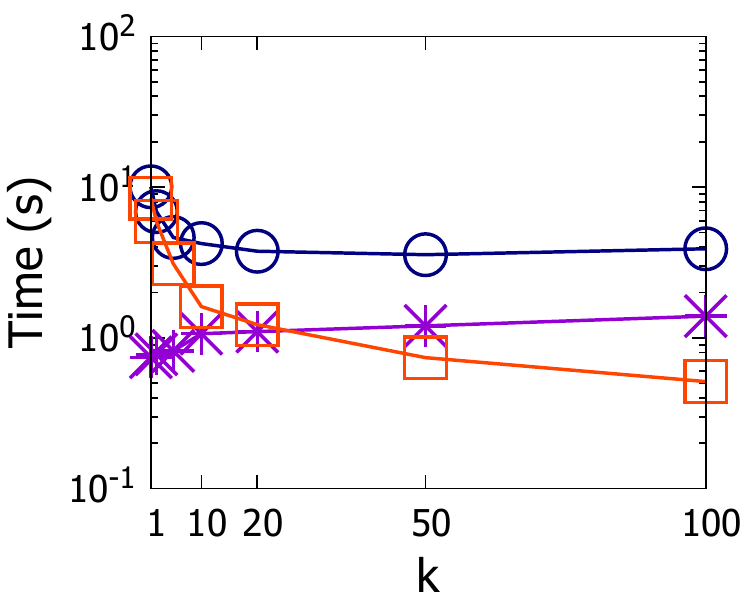}
    &
    \includegraphics[width=0.45\columnwidth]{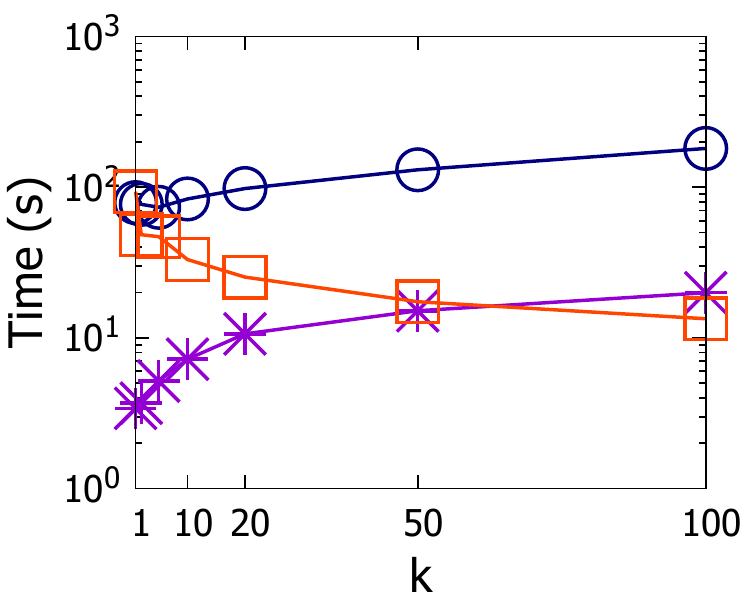}
    \\
    (a) \nethept{} & (b) \dblp{} & (c) \twitter{} & (d) \livejournal{}
  \end{tabular}
  \caption{Execution time vs. $k$ with $\rho=0.1$, \LT{}}\label{fig:p=0.1varyk_time_lt}
\end{figure*}

\stitle{Effect of $\rho$ in \WC{}}
Figure~\ref{fig:k=5varyp_time_wc}, ~\ref{fig:k=5varyp_memory_wc} and ~\ref{fig:k=5varyp_quality_wc} show the execution time, memory consumption and result quality of each approach by varying $\rho$ from $0.1$ to $1.3$, respectively.
The execution time of \RCELF{} outperforms \SG{}, \PMC{}, \DSSA{} and \IMM{} for all $\rho$ values in four datasets, as illustrated in Figure~\ref{fig:k=5varyp_time_wc}\footnote{The missing values are due to the algorithms run out of memory or cannot return the result within 10 hours.}.
However, \IMM{} and \DSSA{} perform as better as \RCELF{} when $\rho=1$ in \twitter{} and \livejournal{} (see Figure~\ref{fig:k=5varyp_time_wc}(c) and (d)).
It confirms our analysis in Section~\ref{sec:eximp}, i.e., \IMM{} and \DSSA{} are pretty good when $\rho=1$. 

We measure the memory consumption of each approach by varying $\rho$ in four datasets in Figure~\ref{fig:k=5varyp_memory_wc}.
Since the space complexity of \RCELF{} is $O(n+m)$, the memory consumption of \RCELF{} performs better than all other competitors by varying $\rho$ from 0.1 to 1.3 in all datasets.
Interestingly, the memory consumptions of \IMM{} and \DSSA{} are rising when $\rho$ varies from 1 to 0.1, and from 1 to 1.3 in two large datasets (see Figure~\ref{fig:k=5varyp_memory_wc}(c) and (d)).
The reason is that the memory consumptions of \IMM{} and \DSSA{} are heavily relying on the prorogation probability of each edge.
Specifically, \IMM{} and \DSSA{} requires more bootstrap iterations when $\rho$ is small. Thus the memory consumptions are rising from $1$ to $0.1$.
When $\rho$ is rising from $1.0$ to $1.3$, the number of nodes in each generated \RR{} set by \IMM{} and \DSSA{} will increase dramatically as the strongly connected properties of the underlying social network.
Thus the size of \RR{} set grows up.

{The result qualities of all approaches by varying $\rho$ in four datasets are illustrated in Figure~\ref{fig:k=5varyp_quality_wc}.
The practical spread results of different approaches are similar as all approaches guarantees the same approximation ratio of result\footnote{We use the author released code of \PMC{} and follow the parameter setting in~\cite{ohsaka2014fast}. However, \PMC{} performs worse than other competitors in all cases.}.}

\stitle{Effect of $k$ in \WC{}}
We test the effect of $k$ by fixing $\rho=0.1$ in all four datasets.
Figure~\ref{fig:p=0.1varyk_time_wc} shows the execution cost of each method, varying $k$ from 1 to 100.
Observe that our proposal \RCELF{} consistently outperforms all competitors.
Specifically, our proposal \RCELF{} is up to 86.1\X{}, 149.7\X{}, 72.1\X{} and 47.8\X{} faster than \IMM{} in \nethept{}, \dblp{}, \twitter{}, and \livejournal{}, respectively.
It is faster than \DSSA{} by 134\X{}, 116\X{}, 86.8\X{} and 57.5\X{}, respectively.
Besides that, {the performance of \PMC{} is worse than that of \SG{}. Although \PMC{} is proposed to improve \SG{}, it costs more time since it is sensitive to the connectivity of the graph. Unfortunately, the graph connectivity is poor with \WC{} model, the overhead of generating DAGs in \PMC{} is larger than its benefits.}

The memory consumptions of each approach in four datasets are illustrated in Figure~\ref{fig:p=0.1varyk_memory_wc}.
\RCELF{}'s memory requirement is stable with regard to $k$ in all datasets as the space complexity of our approach is $O(m+n)$.
It requires the minimum memory space among all competitors.
For example, when $k=5$ \DSSA{}, \IMM{} and \RCELF{} require 1,548MB, 1,625MB and 47MB memory space on \nethept{}, respectively.
In particular, \DSSA{} consumes more than 27GB memory space when $k=1$ to return a solution on a 18MB dataset \dblp{}.

{Figure~\ref{fig:p=0.1varyk_qualty_wc} shows the spread values of different approaches in four datasets by varying $k$ from $1$ to $100$.
As expected, our proposal \RCELF{} performs as well as other competitors (e.g., \IMM{}, \DSSA{}) in all tested settings.
This also verified that \RCELF{} guarantees the approximation ratio of the result as other approximate approaches.}

\stitle{\RCELF{} in \IC{} and \LT{}}
Figure~\ref{fig:p=0.001varyk_time_ic} shows the execution time of \RCELF{}, \SG{}, \PMC{}, \IMM{}, and \DSSA{} by varying $k$ in \IC{} model, where $w(u,v)$ follows the setting of~\cite{cheng2013staticgreedy}, i.e.,  $w(u,v)=0.001$.
\RCELF{} is one to two orders of magnitude faster than \IMM{} and \DSSA{} in \nethept{}, \dblp{} among all four datasets, as shown in Figure~\ref{fig:p=0.001varyk_time_ic}.
In \livejournal{} dataset, 
\IMM{} cannot return results as it incurs extremely large memory consumption for all $k$ values. 
When $k=1,2$ and $5$, \DSSA{} also infeasible due to huge memory consumption,
it confirms that \DSSA{} performs worse when $k$ is small~\cite{huang2017revisiting} (see Figure~\ref{fig:p=0.001varyk_time_ic}(d)).

\SG{} and \PMC{} do not work with \LT{}~\cite{li2018influence}.
Figure~\ref{fig:p=0.1varyk_time_lt} shows the execution time of \RCELF{}, \IMM{}, and \DSSA{} by varying $k$ in \LT{} model.
\RCELF{} is better than or at least comparable with \IMM{} and \DSSA{} in all four datasets.

For the sake of presentation,  we omit the memory consumption and result quality results as they are similar with Figure~\ref{fig:p=0.1varyk_memory_wc} and ~\ref{fig:p=0.1varyk_qualty_wc}, respectively.

\stitle{Scalability of \RCELF{}}
We verify the scalability of \RCELF{} in the largest dataset (i.e., \twitterrv{}) used in literature in \WC{} model.
Figure~\ref{fig:k=5varyp_time_large}(a) shows the execution cost of \RCELF{}, \IMM{}, and \DSSA{} by varying $\rho$ with $k=5$.
\RCELF{} outperforms other competitors in all cases.
We measure the execution time of \RCELF{}, \IMM{} and \DSSA{} by varying $k$ with $\rho=0.1$ in Figure~\ref{fig:k=5varyp_time_large}(b).
When $k$ is small, \RCELF{} performs better than \IMM{} and \DSSA{}.
When $k$ is large, \DSSA{} outperforms \RCELF{} due to its superiority for large $k$.
Besides, the memory consumption of \RCELF{} is less than \IMM{} and \DSSA{} as \RCELF{} does not incur any extra memory overhead.

\begin{figure}
	\centering
	\begin{tabular}{cc}
		\includegraphics[width=0.45\columnwidth]{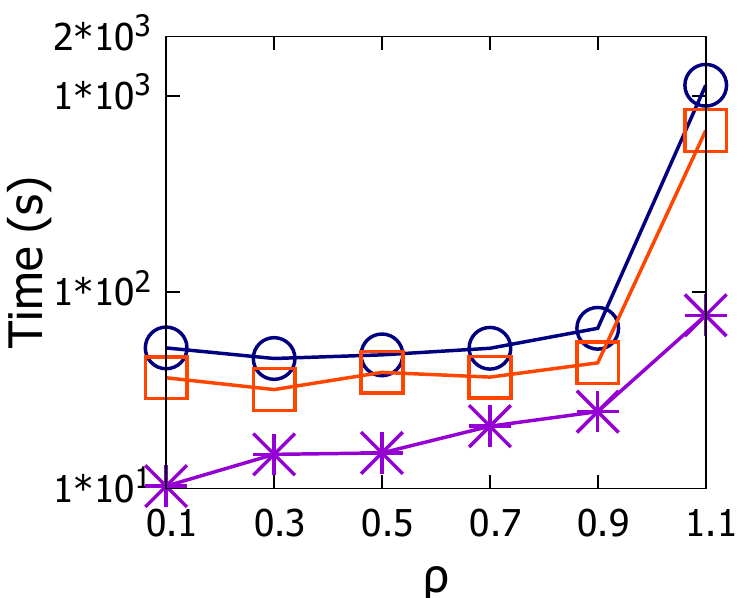}
		&
		\includegraphics[width=0.45\columnwidth]{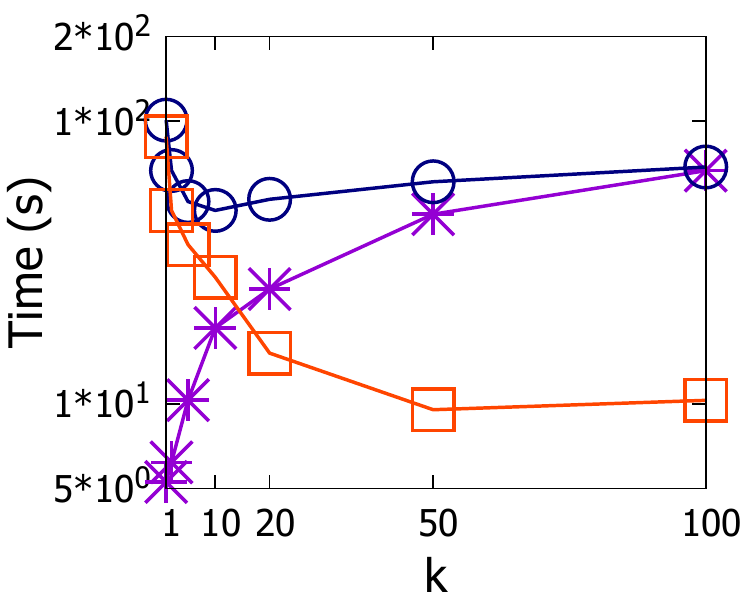}
		\\
		(a) Varying $\rho$, $k=5$ & (b) Varying $k$, $\rho=0.1$
	\end{tabular}
	\caption{Execution time, \twitterrv{}, \WC{}}\label{fig:k=5varyp_time_large}
\end{figure}

\stitle{Optimization Evaluation}
Here we evaluate the effectiveness of our proposed optimization techniques (i.e., Lemma \ref{lem:sigma}).
We evaluate the effect of Lemma~\ref{lem:sigma} with $k=100$ and $\rho=0.1$.
As illustrated in Figure~\ref{fig:lemma}, our Lemma~\ref{lem:sigma} optimization offer saving of 1.0\%, 26.7\%, 21.7\%, 40.9\% (compared to \RCELF{} without Lemma~\ref{lem:sigma}) in \nethept{}, \dblp{}, \twitter{} and \livejournal{}, respectively.

\begin{figure}
	\centering
		\includegraphics[width=0.50\columnwidth]{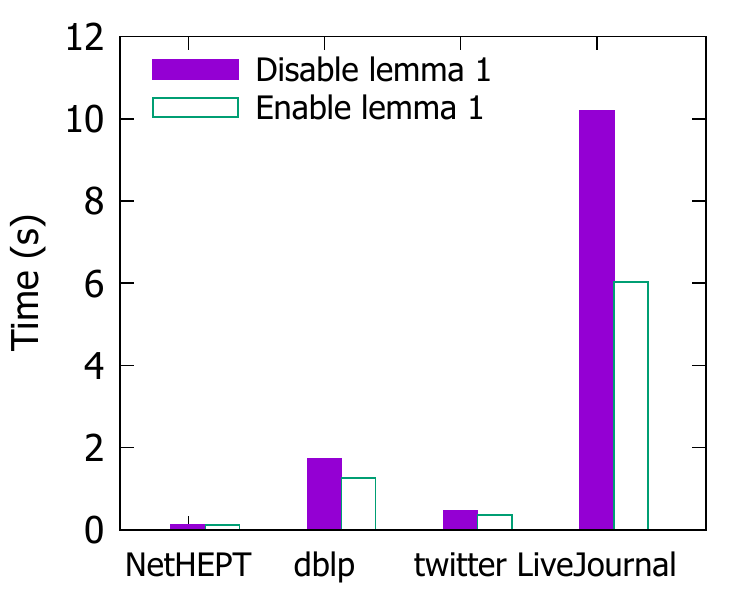} \\
	\caption{Effect of Lemma~\ref{lem:sigma}}\label{fig:lemma}
\end{figure}

\trim
\section{Conclusion}~\label{sec:con}
In this paper, we discuss existing approximation solutions for \IMP{},
which are compromising time efficiency or memory consumption for the approximate result quality.
In order to address that, we propose a residual-based algorithm \RCELF{} for \IMP{}, which achieves good time efficiency, low memory consumption and approximate guaranteed result quality concurrently in generalized \IC{} and \LT{} models.
Besides, we propose several optimizations to accelerate the performance of \RCELF{}.
We demonstrate the superiority of \RCELF{} on standard real benchmarks.
We plan to extend our \RCELF{} to Triggering model and Time Aware model in future work.

\begin{small}
\bibliographystyle{IEEEtran}
\bibliography{ref}
\end{small}

\end{document}